\documentclass[12pt, draftclsnofoot, onecolumn]{IEEEtran}
\usepackage{subfigure}
\usepackage{setspace}
\usepackage{multirow} 
\usepackage{algorithm}
\usepackage{algorithmic}
\usepackage{amsmath}
\usepackage{amssymb}
\usepackage{latexsym}
\usepackage{multirow}
\usepackage{epsfig}
\usepackage{graphics}
\usepackage{graphicx}
\usepackage{mathrsfs}
\usepackage{subfigure}
\usepackage{slashbox}
\usepackage{mathrsfs}
\usepackage{bbding}
\usepackage{cite}
\usepackage{color}
\usepackage{xcolor}
\usepackage{booktabs}
\usepackage{amssymb,mathrsfs,amsmath}
\newcommand{\bm}[1]{\mbox{\boldmath{$#1$}}}
\usepackage{empheq}
\usepackage{microtype}

\usepackage{amsthm}
\theoremstyle{definition}

\newtheorem{theo}{Theorem}
\newtheorem{rem}{Remark}

\usepackage{hyperref}

\allowdisplaybreaks [4]

\begin{document}

\title{Beamforming Optimization for Active Intelligent Reflecting Surface-Aided SWIPT}
\author{Ying~Gao, Qingqing~Wu, Guangchi~Zhang, Wen~Chen, Derrick~Wing~Kwan~Ng, and Marco~Di~Renzo\vspace{-14.8mm}
\thanks{Y.~Gao and Q.~Wu are with the State Key Laboratory of Internet of Things
for Smart City, University of Macau, Macao 999078, China (e-mail: yinggao@um.edu.mo; qingqingwu@um.edu.mo). G.~Zhang is with the School of Information Engineering, Guangdong University of Technology, Guangzhou 510006, China  (e-mail: gczhang@gdut.edu.cn). W.~Chen is with
the Department of Electronic Engineering, Shanghai Jiao Tong University,
Shanghai 201210, China (e-mail: wenchen@sjtu.edu.cn). D.~W.~K.~Ng is with the School of Electrical Engineering and Telecommunications, University of New South Wales, NSW 2052, Australia (e-mail: w.k.ng@unsw.edu.au). M.~Di~Renzo is with Universit\'e Paris-Saclay, CNRS, CentraleSup\'elec, Laboratoire des Signaux et Syst\`emes, 3 Rue Joliot-Curie, 91192 Gif-sur-Yvette, France. (marco.di-renzo@universite-paris-saclay.fr). }}

\maketitle

\begin{abstract}
	\vspace{-1.6mm}
	\emph{Active} intelligent reflecting surface (IRS) has been recently proposed to alleviate the product path loss attenuation inherent in the IRS-aided cascaded channel. 
	In this paper, we study an active IRS-aided simultaneous wireless information and power transfer (SWIPT) system. Specifically, an active IRS is deployed to assist a multi-antenna access point (AP) to convey information and energy simultaneously to multiple single-antenna information users (IUs) and energy users (EUs). Two joint transmit and reflect beamforming optimization problems are investigated with different practical objectives. The first problem maximizes the weighted sum-power harvested by the EUs subject to individual signal-to-interference-plus-noise ratio (SINR) constraints at the IUs, while the second problem maximizes the weighted sum-rate of the IUs subject to individual energy harvesting (EH) constraints at the EUs. The optimization problems are non-convex and difficult to solve optimally. To tackle these two problems, we first rigorously prove that dedicated energy beams are not required for their corresponding semidefinite relaxation (SDR) reformulations and the SDR is tight for the first problem, thus greatly simplifying the AP precoding design. Then, by capitalizing on the techniques of alternating optimization (AO), SDR, and successive convex approximation (SCA), computationally efficient algorithms are developed to obtain suboptimal solutions of the resulting optimization problems. Simulation results demonstrate that, given the same total system power budget, significant performance gains in terms of operating range of wireless power transfer (WPT), total harvested energy, as well as achievable rate can be obtained by our proposed designs over benchmark schemes (especially the one adopting a passive IRS). Moreover, it is advisable to deploy an active IRS in the proximity of the users for the effective operation of WPT/SWIPT.
\end{abstract}

\vspace{-1mm}
\begin{IEEEkeywords}
\vspace{-2mm}
Active intelligent reflecting surface (IRS), simultaneous wireless information and power transfer (SWIPT), beamforming optimization, radio frequency-based energy harvesting (EH). 
\end{IEEEkeywords}

\section{Introduction}
With the increasing number and diversification of intelligent devices in Internet-of-Things (IoT) networks, emerging applications such as smart cities, mobile streaming media, and multisensory virtual reality become possible. One of the main challenges in the roll-out of reliable IoT is the energy limitation of battery-powered devices. Against this background, the dual use of radio frequency (RF) signals for enabling simultaneous wireless information and power transfer (SWIPT) has attracted intense interest \cite{2022_Qingqing_WIT_WET_overview}. However, energy users (EUs) and information users (IUs) typically operate with very different power sensitivity requirements (e.g., $-10$ dBm for EUs versus $-60$ dBm for IUs) \cite{2013_Jie_SWIPT_arxiv}. As such, the operating range of wireless power transfer (WPT) is fundamentally limited by the severe path loss over long signal propagation distances, which restricts the performance of SWIPT systems. Although numerous multiple-input multiple-output (MIMO) technologies can considerably improve the efficiency of both WPT and wireless information transmission (WIT), their practical implementations are still hindered by the required exceedingly high energy consumption and hardware cost \cite{2017_Shunqing_Intro}. 

Recently, intelligent reflecting surface (IRS) or reconfigurable intelligent surface (RIS), comprised of a large number of passive metamaterial elements, has emerged as a cost- and energy-efficient solution to unlock the potential of future wireless networks \cite{2020_Qingqing_IRS_Intro,2019_Qingqing_Joint,2020_Marco_RIS_vsRelaying,2020_Marco_RIS_JSAC,2021_Shun_AIRIS}. Specifically, by smartly adapting the phase shifts of all the IRS elements according to the time-varying environment, an IRS is capable of reconfiguring the wireless propagation channels for enhancing the desired signal strength and/or mitigating interference, thus improving the communication performance. It was firstly revealed in \cite{2019_Qingqing_Joint} that when the number of IRS elements, denoted by $N$, is large enough but finite, the receive signal power or the signal-to-noise ratio (SNR) gain scales with $\mathcal O(N^2)$ asymptotically. This pioneering result has then inspired major interests in studying the joint active and passive beamforming design for various system setups (see, e.g., \cite{2020_Qingqing_Discrete,2019_Miao_Secure,2020_Sixian_UAV,2020_Xinrong_AN,2021_Qian_SNR,2021_Marco_Rate,2021_Meng_CoMP,2021_Jiangbo_covert,2021_Xiaowei_IRS_UAV,2022_Xiaowei_IRS_UAV}). While the works \cite{2019_Qingqing_Joint,2020_Qingqing_Discrete,2019_Miao_Secure,2020_Sixian_UAV,2020_Xinrong_AN,2021_Qian_SNR,2021_Marco_Rate,2021_Meng_CoMP,2021_Jiangbo_covert,2021_Xiaowei_IRS_UAV,2022_Xiaowei_IRS_UAV} mainly focus on exploiting IRSs for effective WIT, the high passive beamforming gain promised by IRSs is also attractive for WPT \cite{2022_Qingqing_WIT_WET_overview,2020_Qingqing_IRS_Intro}. To reap this benefit, one line of research investigated IRS-aided wireless-powered communication networks (WPCNs), aiming at improving the communication performance by leveraging an IRS to assist WPT and WIT across different time slots \cite{2021_Dingcai_HyNOMA_TDMA,2022_Qingqing_WPCN,2022_Meng_WPCN,2022_zhendong_WPCN_robust}. On the other hand, another line of research focused on exploiting the high passive beamforming gain to enlarge the rate-energy region of SWIPT systems offered by IRSs \cite{2020_Qingqing_SWIPT_letter,2020_Qingqing_SWIPT_QoS,2020_Cunhua_SWIPT,2021_Shiqi_SWIPT_Discrete,2021_Shayan_SWIPT,2022_Dongfang_SWIPT}. For example, the authors of \cite{2020_Qingqing_SWIPT_letter} jointly optimized the transmit precoder at the access point (AP) and the phase shifts at the IRS to maximize the weighted sum-power of the EUs, while satisfying the minimum signal-to-interference-plus-noise ratio (SINR) requirements of the IUs. Inspired by \cite{2020_Qingqing_SWIPT_letter}, the authors of \cite{2020_Qingqing_SWIPT_QoS} studied the transmit power minimization subject to the individual quality-of-service (QoS) constraints at both the IUs and the EUs. Besides, the weighted sum-rate of all the IUs was maximized in \cite{2020_Cunhua_SWIPT}, where the weighted sum-power harvested by all the EUs is ensured to be higher than a predefined value for QoS provisioning. 

In spite of the appealing advantages of IRSs, the performance of passive IRS-aided systems may not offer a wide coverage extension because of the product path loss attenuation law, unless the number of IRS elements is very large \cite{2021_Qingqing_Tutorial}. Particularly, the end-to-end path loss of the transmitter-IRS-receiver link is generally significantly more severe than that of the unobstructed direct link, since the former is the product of the path losses of the transmitter-IRS and the IRS-receiver links. To circumvent this problem, one may need to install a large number of passive reflecting elements and/or to place the passive IRS in close vicinity to either the transmitter or the receiver, which, however, may not always be practically efficient or even feasible. As a remedy, the concept of \emph{active} IRS has been proposed recently (see, e.g., \cite{2021_Ruizhe_active_SIMO,2021_Zhang_active_6G}) to alleviate the product path loss attenuation law. In particular, an active IRS is generally compromised of a number of active reflecting elements, each of which independently integrates a reflection-type amplifier, e.g., a tunnel diode, and thus can not only alter the incident signals' phases, but also amplify them at the cost of additional low power consumption. Indeed, several papers in the field of metamaterials and communications have proved that an IRS can reflect the incident power with unitary power efficiency, for any angles of incidence and reflection, if local power gains and losses are present along the surface of the IRSs, see, e.g., \cite{2016_Estakhri_metasurfaces,2017_Ana_reflectors,2021_Marco_Electromagnetics}. The realization of these structures is more difficult than the conventional design of locally-passive IRSs, but it usually results in better performance. The surfaces in \cite{2016_Estakhri_metasurfaces,2017_Ana_reflectors,2021_Marco_Electromagnetics} are usually globally passive, i.e., the reflected power is not greater than the incident power. In active IRSs, on the other hand, the reflected power is greater than the incident power. Both, however, assume that local power amplifications can be realized along the surface of the IRS. Recently, some innovative efforts have been devoted to beamforming optimization for active IRS-aided systems \cite{2021_Ruizhe_active_SIMO,2021_Zhang_active_6G,2021_Changsheng_active_or_passive,2022_Kunzan_active_hybrid,2022_Piao_Active,2022_Guangji_Active}. For instance, the authors of \cite{2021_Ruizhe_active_SIMO} studied the SNR maximization problem when introducing an active IRS into a single-input multiple-output (SIMO) system. The simulation results demonstrate that, given the same IRS power budget, an active IRS-aided system outperforms its passive counterpart. In \cite{2021_Zhang_active_6G}, it is shown that an active IRS is capable of achieving noticeable capacity gains regardless of the strength of the direct link. Furthermore, the results in \cite{2021_Changsheng_active_or_passive} indicate that, with the optimized IRS placement, an active IRS has better performance than that of a passive IRS in some practical scenarios. \looseness=-1

Although the above works have validated the superiority of adopting an active IRS over a passive IRS under some system setups, to our best knowledge, the potential performance gain of integrating an active IRS into SWIPT systems remains uninvestigated. Moreover, since the information signals for the IUs can be utilized at the EUs for energy harvesting (EH), some fundamental questions remain to be answered in active IRS-aided SWIPT systems. First, are dedicated energy beams required to maximize the weighted sum-power of the EUs with the consideration of information transmission? This question is motivated by the result in \cite{2020_Qingqing_SWIPT_letter}, which shows that dedicated energy beams are not required for a passive IRS-aided SWIPT system, generalizing the finding in \cite{2013_Jie_SWIPT_arxiv} to the case with arbitrary user channels. For active IRS-assisted systems, however, the problem formulation of beamforming design is rather different due to the newly imposed amplification power constraint and the non-negligible IRS-amplified noise power. Thus, it is unknown whether the conclusions drawn in \cite{2020_Qingqing_SWIPT_letter} still hold for active IRS-aided SWIPT systems. The second open question is: is sending only information beams sufficient to satisfy the individual EH constraints at the EUs while maximizing the achievable weighted sum-rate of the IUs? In other words, can dedicated energy beams be removed or set to zero when solving the QoS-constrained weighted sum-rate maximization problem? In \cite{2020_Cunhua_SWIPT}, specially, the authors studied such a problem for a passive IRS-aided SWIPT system, by simply assuming that there is no energy beamforming applied at the AP. Thus, the aforementioned fundamental issues remain unsolved. \looseness=-1

Motivated by these considerations, we investigate an active IRS-aided SWIPT system where an active IRS is deployed to assist the information/power transfer from a multi-antenna AP to multiple single-antenna IUs and EUs, as shown in Fig. \ref{Fig:system model}. The transmit precoder at the AP and the reflection-coefficient matrix at the IRS are jointly optimized by considering two different design criteria. In particular, the first problem maximizes the weighted sum-power harvested by the EUs subject to individual SINR constraints at the IUs, while the second problem maximizes the weighted sum-rate of the IUs subject to individual EH constraints at the EUs. The main contributions of this paper are summarized as follows.
\begin{itemize}
	\item To obtain useful insights about active IRS-aided WPT systems, we first consider a special case of the weighted sum-power maximization problem where there exist no IUs and apply the alternating optimization (AO), semidefinite relaxation (SDR), and successive convex approximation (SCA) techniques to obtain a suboptimal solution. Next, for the general case where the EUs and the IUs coexist, we rigorously prove that, despite the presence of a new amplification power constraint and non-negligible IRS-amplified noise power, dedicated energy beams are not required, which greatly simplifies the AP precoding design. Exploiting the obtained insight, we then propose a computationally efficient algorithm to solve the resulting problem suboptimally. 
	\item For the weighted sum-rate maximization problem, we first unveil that dedicated energy beams are not needed for its SDR reformulation. Additionally, although the tightness of the SDR cannot be confirmed, a high-quality suboptimal solution for the original problem can be constructed from the optimal solution of the reformulated SDR problem. Building upon these insights, we consider the SDR reformulation with no dedicated energy beams to simplify the AP precoding design. Subsequently, we propose a computationally efficient suboptimal algorithm based on the AO and SCA techniques for the resulting problem and show how to approximately recover the transmit precoder if the obtained transmit beamforming matrices are not rank-one. 
	\item Numerical results demonstrate that by introducing an active IRS, the performance of SWIPT systems can be significantly enhanced in terms of operating range of WPT, total harvested energy, as well as achievable rate, as compared to passive IRS-aided SWIPT systems, under the assumption that the total system power budgets are the same. Furthermore, it is shown that the deployment of an active IRS close to the users is beneficial for WPT/SWIPT systems. 
\end{itemize}

The remainder of this paper is organized as follows. Section \ref{Section_model and formu} introduces the active IRS-aided SWIPT system model and presents the formulations of the weighted sum-power and sum-rate maximization problems. In Sections \ref{Section_solution to (P1)} and \ref{Section_solution to (P2)}, we propose efficient algorithms for solving the two formulated problems, respectively. Numerical results are presented in Section \ref{Section_simulation} to evaluate the performance of the proposed algorithms. Finally, Section \ref{Section_conclusion} concludes the paper. 

\emph{Notations:} Scalars, vectors, and matrices are denoted in lower-case, boldface lower-case, and boldface upper-case letters, respectively. $\mathbb C^{x\times y}$ denotes the space of $x\times y$ complex-valued matrices. $\mathbb H^N$ represents the set of all $N$-dimensional complex Hermitian matrices. $\mathbb E(\cdot)$ denotes the statistical expectation. The distribution of a circularly symmetric complex Gaussian (CSCG) random vector with a mean vector $\bm x$ and a covariance matrix $\mathbf \Sigma$ is denoted by $\mathcal {CN}\left(\bm x, \mathbf \Sigma\right)$. $\sim$ and $\triangleq$ stand for ``distributed as'' and ``defined as'', respectively. 
$\jmath$ denotes the imaginary unit, i.e., $\jmath^2 = 1$. The phase and real part of a complex number are denoted by ${\rm arg}(\cdot)$ and ${\rm Re}\{\cdot\}$, respectively. For a vector $\bm a$, $\left\|\bm a\right\|$ and $[\bm a]_{n}$ represent its Euclidean norm and $n$-th element, respectively. diag$(\bm a)$ denotes a diagonal matrix with each diagonal element being the corresponding element in $\bm a$. $\mathbf I_N$ is an identity matrix of size $N\times N$; $\mathbf 0$ and $\mathbf 1$ denote an all-zero matrix and an all-one matrix, respectively, with dimensions determined from the context. For a square matrix $\bm S$, ${\rm tr}(\bm S)$ represents its trace; $\bm S  \succeq \bm 0$ indicates that $\bm S$ is positive semidefinite. For a matrix $\bm A$ of arbitrary size, $\left\|\bm A\right\|_F$, ${\rm rank}(\bm A)$ and $[\bm A]_{i,j}$ denote its Frobenius norm, rank and $(i,j)$-th element, respectively. $(\cdot)^H$ corresponds to the conjugate transport of a vector or matrix. $ \odot$ denotes the Hadamard product. For a set $\mathcal K$, $\left|\mathcal K\right| $ denotes its cardinality. $\mathcal O(\cdot)$ expresses the big-O notation. 

\begin{figure}[!t]
	\vspace{-3mm}
	\centering
	\includegraphics[width=0.6\textwidth]{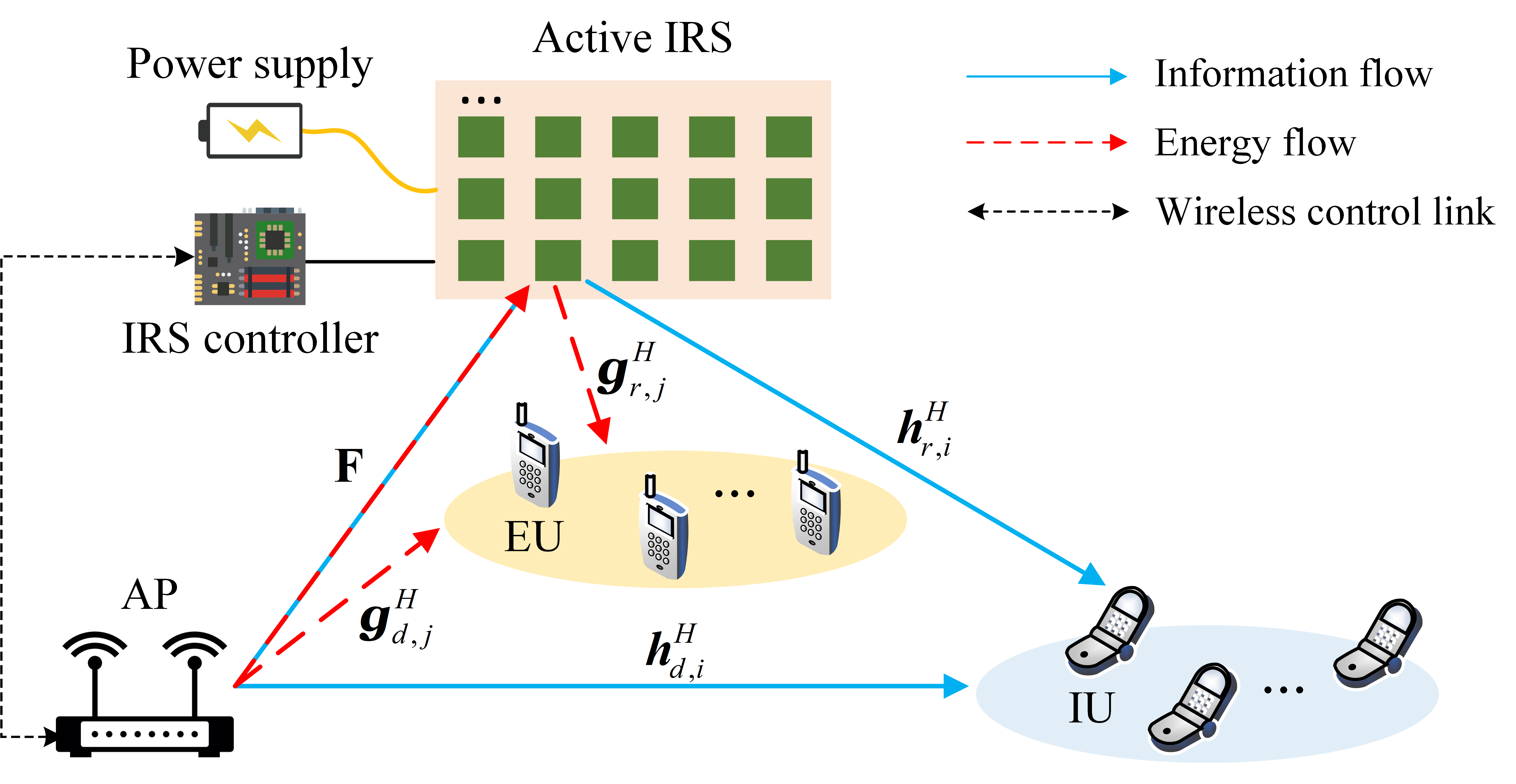}
	\caption{An active IRS-aided SWIPT system. } \label{Fig:system model}
	\vspace{-5mm}
\end{figure}

\section{System Model and Problem Formulation}\label{Section_model and formu}

\subsection{System Model}
As shown in Fig. \ref{Fig:system model}, we consider an active IRS-aided SWIPT system, which comprises an AP with $M$ antennas, an active IRS\footnote{Note that the revealed results and the proposed algorithms in this paper can be extended to the case with a \emph{hybrid} IRS \cite{2022_Nguyen_hybrid} after slight modifications.} with $N$ reflecting elements, and two sets of single-antenna users, i.e., $K_{\rm I}$ IUs and $K_{\rm E}$ EUs, denoted by $\mathcal {K_I} = \left\lbrace1, \cdots, K_{\rm I} \right\rbrace $ and $\mathcal {K_E} = \left\lbrace1, \cdots, K_{\rm E} \right\rbrace $, respectively. In particular, the active IRS is supported by an external power supply and each of its elements can not only alter the incident signals' phases, but also amplify the incident signals with an integrated reflection-type amplifier. Also, at the AP, we consider a linear transmit precoding for SWIPT with $\bm w_i \in \mathbb C^{M\times 1}$ and $\bm v_j \in \mathbb C^{M\times 1}$ denoting the beamforming vectors for IU $i$ and EU $j$, respectively. Hence, the transmitted signal from the AP can be expressed as 
\begin{align}\label{transmitted signal}
\bm x = \sum_{i\in\mathcal {K_I}}\bm w_i s_i^{\rm I} + \sum_{j\in\mathcal {K_E}}\bm v_j s_j^{\rm E},
\end{align}
where $s_i^{\rm I} \in \mathbb C$ and $ s_j^{\rm E}\in \mathbb C$ are the information-bearing signal for IU $i$ and the energy-carrying signal for EU $j$, respectively, satisfying $s_i^{\rm I} \sim\mathcal {CN}\left(0, 1\right) $, $\forall  i\in\mathcal {K_I}$ and $\mathbb E\left( \left|s_j^{\rm E} \right|^2 \right) = 1 $, $\forall j\in\mathcal{K_E}$ \cite{2013_Jie_SWIPT_arxiv}. The two signals are assumed to be independent to each other. Let $ P_{\text A}$ denote the total transmit power budget of the AP. From \eqref{transmitted signal}, we have $\mathbb E\left( \bm x^H\bm x\right) = \sum_{i\in\mathcal {K_I}}\left\| \bm w_i\right\|^2 + \sum_{j\in\mathcal {K_E}}\left\| \bm v_j\right\|^2 \leq P_{\text A}$.

For characterizing the theoretical performance gain brought by an active IRS, we assume a quasi-static fading environment and the channel state information (CSI)\footnote{If the IRS is equipped with active sensors, compressive sensing methods \cite{2021_Taha_channel_sensor, 2022_Schroeder_channel_sensor} can be applied to estimate the channels of the AP-IRS and the IRS-user links, respectively; otherwise, the CSI of individual links can be acquired exploiting some existing methods, e.g., \cite{2020_G.T._Channel_PARAFAC-Based,2021_Wei_Channel_RIS-Empowered}.} of all channels involved is assumed to be acquired perfectly by the AP.\footnote{It is worth mentioning that although this paper focuses on the case of perfect CSI, the results (i.e., Theorems \ref{theo1} and \ref{theo2}) can be generalized to the case of imperfect CSI since the proof of Theorems \ref{theo1} and \ref{theo2} does not rely on the accuracy of the channel estimation. Note that it is of great interest to develop robust beamforming designs for the case of imperfect CSI, which, however, goes beyond the scope of this paper and is left for future work.} Let $\bm h_{d,i}^H \in \mathbb C^{1\times M}$ and $\bm h_{r,i}^H \in \mathbb C^{1\times N}$ denote the baseband equivalent channels from the AP to IU $i$ and from the IRS to IU $i$, respectively. The corresponding channels for EU $j$ are denoted by $\bm g_{d,j}^H \in \mathbb C^{1\times M}$ and $\bm g_{r,j}^H \in \mathbb C^{1\times N}$, respectively, and $\mathbf F \in \mathbb C^{N\times M}$ denotes the channel from the AP to the IRS. Let $\mathbf \Theta = \text{diag}\left(u_1, \cdots, u_N\right) \in\mathbb C^{N\times N}$ denote the reflection-coefficient matrix at the IRS. Particularly, $u_n$ can be written as $u_n = \beta_ne^{\jmath\theta_n}$, $n \in \mathcal N \triangleq \left\lbrace1, \cdots, N \right\rbrace $, where $\beta_n \geq 0$ and $\theta_n \in [0, 2\pi) $ represent the reflection amplitude and phase shift of the $n$-th IRS element, respectively. The signal received at IU $i$ is then given by
\begin{align}
y_i^{\text I} = \underbrace{\bm h_{d,i}^H\bm x}_{\text{direct link}} + \underbrace{\bm h_{r,i}^H\bm \Theta\left(\mathbf F\bm x + \mathbf z\right)}_{\text{reflected link}} + n_i = \left(\bm h_{r,i}^H\mathbf \Theta\mathbf F + \bm h_{d,i}^H \right) \bm x + \bm h_{r,i}^H\mathbf \Theta\mathbf z + n_i, i\in\mathcal {K_I}, 
\end{align}
where $\mathbf z \sim\mathcal {CN}\left(\mathbf0, \sigma_z^2\mathbf I_N \right)$ and $n_i \sim\mathcal {CN} \left(0, \sigma_i^2 \right)$ denote the noise introduced by the active IRS\cite{2021_Ruizhe_active_SIMO} and the additive white Gaussian noise (AWGN) at IU $i$, respectively, with $\sigma_z^2$ and $\sigma_i^2$ being the corresponding noise variances. By assuming that the interference caused by the energy signals cannot be cancelled by the IUs, the SINR at IU $i$ can be written as  
\begin{align}\label{equ:SINR_expression}
\text{SINR}_i = \frac{\left|\bm h_i^H \bm w_i \right|^2 }{\sum_{k\in \mathcal {K_I}\backslash\{i\}}\left|\bm h_i^H \bm w_k \right|^2 + \sum_{j\in \mathcal {K_E}}\left|\bm h_i^H \bm v_j \right|^2 + \sigma_z^2\left\|\bm h_{r,i}^H\mathbf \Theta\right\|^2 + \sigma_i^2}, i\in\mathcal {K_I},
\end{align}
where $\bm h_i^H \triangleq \bm h_{r,i}^H\mathbf \Theta\mathbf F + \bm h_{d,i}^H$ denotes the equivalent end-to-end channel from the AP to IU $i$. Accordingly, the achievable rate at IU $i$ in bits/second/Hz (bps/Hz) is given by $R_i = \log_2\left(1 + \text{SINR}_i \right) $. On the other hand, the received RF power at EU $j$, denoted by $Q_j$, can be expressed as
\begin{align}\label{equ:Q_j}
Q_j = \sum_{i\in \mathcal {K_I}}\left|\bm g_j^H \bm w_i \right|^2 + \sum_{m\in \mathcal {K_E}}\left|\bm g_j^H \bm v_m \right|^2 + \sigma_z^2\left\|\bm g_{r,j}^H\mathbf \Theta\right\|^2,  j\in\mathcal {K_E},
\end{align}
where $\bm g_j^H \triangleq \bm g_{r,j}^H\mathbf \Theta\mathbf F + \bm g_{d,j}^H$ denotes the equivalent end-to-end channel from the AP to EU $j$.\footnote{For simplicity, we adopt a linear EH model which is widely used in existing works such as \cite{2021_Dingcai_HyNOMA_TDMA,2022_Qingqing_WPCN,2022_Meng_WPCN,2020_Qingqing_SWIPT_letter,2020_Qingqing_SWIPT_QoS,2020_Cunhua_SWIPT}. The extension to the case with a non-linear EH model \cite{2021_Shayan_SWIPT} is discussed in Remarks \ref{rem3} and \ref{rem5}. Besides, the power of the antenna noise is ignored in EH while that of the noise introduced by the active IRS is considered. This is because the former is generally a negligible constant, while the latter could be non-negligible, especially when the number of the reflecting elements, $N$, is sufficiently large.} 

Suppose that the active IRS is endowed with a maximum amplification power budget $P_{\text{I}}$. Then, we have $\sum_{i\in\mathcal {K_I}}\left\|\mathbf \Theta\mathbf F\bm w_i\right\|^2 + \sum_{j\in\mathcal {K_E}}\left\|\mathbf \Theta\mathbf F\bm v_j\right\|^2 + \sigma_z^2\left\|\mathbf \Theta\right\|_F^2 \leq P_{\text{I}}$. 

\subsection{Problem Formulation}
In this paper, two joint transmit and reflect beamforming optimization problems are considered aiming at two different design criteria. First, we aim to maximize the weighted sum-power received at the EUs while satisfying individual SINR constraints at the IUs, given by $\gamma_i$, $i\in\mathcal {K_I}$. From \eqref{equ:Q_j}, the weighted sum-power received at the EUs can be written as
\begin{align}
\sum_{j\in\mathcal {K_E}}\alpha_j Q_j = \sum_{i\in\mathcal {K_I}}\bm w_i^H\bm S\bm w_i + \sum_{j\in\mathcal {K_E}} \bm v_j^H\bm S\bm v_j + \sum_{j\in\mathcal {K_E}}\alpha_j\sigma_z^2\left\|\bm g_{r,j}^H\mathbf \Theta\right\|^2,
\end{align}
where $\bm S \triangleq \sum_{j\in\mathcal {K_E}}\alpha_j\bm g_j\bm g_j^H $ with $\alpha_j \geq 0$ denoting the given energy weight for EU $j$. Specifically, the larger the value of $\alpha_j$, the higher priority to EU $j$ for EH. Accordingly, the problem of interest can be formulated as
\begin{subequations}
	\begin{eqnarray}
	\hspace{-1cm}\text{(P1)}: &\underset{\left\lbrace \boldsymbol w_i\right\rbrace, \left\lbrace \boldsymbol v_j\right\rbrace,\mathbf \Theta}{\max}& \sum_{i\in\mathcal {K_I}}\bm w_i^H\bm S\bm w_i + \sum_{j\in\mathcal {K_E}} \bm v_j^H\bm S\bm v_j + \sum_{j\in\mathcal {K_E}}\alpha_j\sigma_z^2\left\|\bm g_{r,j}^H\mathbf \Theta\right\|^2 \\
	&\text{s.t.}&  \text{SINR}_i \geq \gamma_i, \forall i\in\mathcal {K_I}, \\
	&&  \sum_{i\in\mathcal {K_I}}\left\|\bm w_i \right\|^2 + \sum_{j\in\mathcal {K_E}}\left\|\bm v_j \right\|^2 \leq P_{\text{A}}, \label{cons_PA}\\
	&&  \sum_{i\in\mathcal {K_I}}\left\|\mathbf \Theta\mathbf F\bm w_i\right\|^2 + \sum_{j\in\mathcal {K_E}}\left\|\mathbf \Theta\mathbf F\bm v_j\right\|^2 + \sigma_z^2\left\|\mathbf \Theta\right\|_F^2 \leq P_{\text{I}}. \label{cons_PI}
	\end{eqnarray}
\end{subequations} 
Problem (P1) is applicable to the scenarios where the IUs have stringent SINR requirements (e.g., delay-limited transmission) while the EUs only require opportunistic EH. 

In addition to (P1), we are also interested in maximizing the weighted sum-rate of the IUs subject to individual EH constraints at the EUs, given by $E_j$, $j\in\mathcal {K_E}$. Let $\mu_i$ denote the weighting factor that controls the access priority of IU $i$. Then, we formulate the optimization problem as follows \looseness=-1
\begin{subequations}
	\begin{eqnarray}
	\text{(P2)}: &\underset{\left\lbrace \boldsymbol w_i\right\rbrace, \left\lbrace \boldsymbol v_j\right\rbrace,\mathbf \Theta}{\max}& \sum_{i\in\mathcal {K_I}}\mu_i \log_2\left( 1 + \text{SINR}_i\right)  \\
	&\text{s.t.}&  \sum_{i\in \mathcal {K_I}}\left|\bm g_j^H \bm w_i \right|^2 + \sum_{m\in \mathcal {K_E}}\left|\bm g_j^H \bm v_m \right|^2 + \sigma_z^2\left\|\bm g_{r,j}^H\mathbf \Theta\right\|^2 \geq E_j, \forall j\in\mathcal {K_E}, \label{P2_cons:E}\\
	&& \eqref{cons_PA}, \eqref{cons_PI}.
	\end{eqnarray}
\end{subequations}  
In contrast to (P1), (P2) applies to the scenarios where the EUs impose strict requirements on EH while the IUs have relaxed constraints for information transmission. 

Since $\{\bm w_i\}$, $\{\bm v_j\}$, and $\mathbf \Theta$ are intricately coupled in the $\text{SINR}_i$ of IU $i$ given in \eqref{equ:SINR_expression}, the received power $Q_j$ of EU $j$ given in \eqref{equ:Q_j}, and the amplification power constraint given in \eqref{cons_PI}, optimization problems (P1) and (P2) are both non-convex, and are hence challenging to solve optimally. Moreover, for the considered active IRS-assisted SWIPT system, it remains unknown whether introducing dedicated energy beams $\{\bm v_j\}$ is necessary for achieving the optima of (P1) and (P2) in the presence of information beams $\{\bm w_i\}$. 

\section{Proposed Solution to Problem (P1)}\label{Section_solution to (P1)}
In this section, we first consider and solve a special case of (P1) with no IUs to gain useful insights into active IRS-aided WPT systems. Then, we answer the question of whether the optimal solution to (P1) requires dedicated energy beams in the general case where the IUs and the EUs coexist. Finally, an efficient algorithm is proposed. 

\subsection{Special Case with No IUs}\label{subsec_WPT}
By setting $\bm w_i = \mathbf 0$, $\gamma_i = 0$, $\forall i\in\mathcal {K_I}$, (P1) is simplified to 
\begin{subequations}
	\begin{eqnarray}
	\hspace{-1.5cm}\text{(P1-NoIUs)}: &\underset{\left\lbrace \boldsymbol v_j\right\rbrace,\mathbf \Theta}{\max}& \sum_{j\in\mathcal {K_E}} \bm v_j^H\bm S\bm v_j + \sum_{j\in\mathcal {K_E}}\alpha_j\sigma_z^2\left\|\bm g_{r,j}^H\mathbf \Theta\right\|^2  \\
	&\text{s.t.}&  \sum_{j\in\mathcal {K_E}}\left\|\bm v_j \right\|^2 \leq P_{\text{A}},\\
	&& \sum_{j\in\mathcal {K_E}}\left\|\mathbf \Theta\mathbf F\bm v_j\right\|^2 + \sigma_z^2\left\|\mathbf \Theta\right\|_F^2 \leq P_{\text{I}}. \label{WPT_cons:amp}
	\end{eqnarray}
\end{subequations} 
The coupling between $\{\bm v_j\}$ and $\bm \Theta$ in both the objective function and constraint \eqref{WPT_cons:amp} introduces non-convexities to (P1-NoIUs). To tackle this issue, we apply the AO method to iteratively optimize $\{\bm v_j\}$ and $\bm \Theta$ until convergence is reached, similar to \cite{2020_Qingqing_SWIPT_letter} but with the exceptions given below. \looseness=-1  

\subsubsection{Optimizing $\{\bm v_j\}$} For any given $\mathbf \Theta$, it was shown in \cite{2020_Qingqing_SWIPT_letter} that, without constraint \eqref{WPT_cons:amp}, employing one dedicated energy beam is sufficient to achieve the optimality and the optimal energy precoder should be aligned to the dominant eigenvector of $\bm S$, denoted by $\bm {v_S}$. However, with the newly introduced constraint \eqref{WPT_cons:amp}, whether employing one dedicated energy beam is sufficient to achieve the optimality is unknown, while aligning the energy precoder to $\bm {v_S}$ may help to improve the objective value to a certain extent but yield a highly suboptimal solution. Therefore, the result in \cite{2020_Qingqing_SWIPT_letter} is not directly applicable to the considered problem. In the following, we solve the subproblem with respect to $\{\bm v_j\}$ by applying the SDR technique \cite{2022_Kaitao_UAV_letter} and answer the aforementioned question. Specifically, we define $\bm W_{\mathrm E} = \sum_{j\in \mathcal {K_E}}\bm v_j\bm v_j^H$, which needs to satisfy $\bm W_{\mathrm E} \succeq \bm 0$ and ${\rm rank}\left(\bm W_{\mathrm E}\right) \leq \min\left(M, K_{\rm E} \right)$. By relaxing the rank constraint on $\bm W_{\mathrm E}$, the subproblem can be recast as (with some constant terms ignored)
\begin{subequations} \label{WPT_sub1_SDR}
	\begin{eqnarray}
	&\underset{\boldsymbol W_{\mathrm E}\in\mathbb H^M}{\max}&  \text{tr}\left(\bm S \bm W_{\mathrm E}\right)   \\
	&\text{s.t.}&  {\rm tr}\left(\bm W_{\mathrm E} \right) \leq P_{\text{A}}, \ {\rm tr}\left(\bm C\bm W_{\mathrm E} \right)  \leq \bar P_{\text{I}}, \ \bm W_{\mathrm E} \succeq \bm 0,
	\end{eqnarray}
\end{subequations}
where $\bm C \triangleq\mathbf F^H\bm \Theta^H \bm \Theta\mathbf F$ and $\bar P_{\text{I}} \triangleq P_{\text{I}} - \sigma_z^2\left\|\mathbf \Theta\right\|_F^2$. As problem \eqref{WPT_sub1_SDR} is a convex semidefinite program (SDP), it can be optimally solved by existing convex optimization solvers, e.g., CVX \cite{2004_S.Boyd_cvx}. Regarding the rank of the obtained optimal solution, we have the following remark. 
\vspace{-3mm}
\begin{rem}\label{rem1}
According to \cite[Theorem 3.2]{2010_Yongwei_Rank}, there exists an optimal solution $\bm W_{\mathrm E}^*$ to problem \eqref{WPT_sub1_SDR} such that $\left( {\rm rank}\left(\bm W_{\mathrm E}^* \right)\right)^2 \leq 2$, where $2$ corresponds to the number of linear constraints in \eqref{WPT_sub1_SDR}.  Moreover, if $P_{\rm A} > 0$ and $\bar P_{\rm I} > 0$, $\bm W_{\mathrm E}^* = \mathbf 0$ cannot be the optimal solution. Hence, an optimal solution that fulfills $ {\rm rank}\left(\bm W_{\mathrm E}^* \right) = 1$ should exist for problem \eqref{WPT_sub1_SDR}. 
\end{rem}
\vspace{-3mm}
Although Remark \ref{rem1} only indicates the existence of a rank-one optimal solution, the rank-reduction techniques in \cite{2010_Yongwei_Rank} can always be applied to construct a rank-one optimal solution from its non-rank-one optimal counterpart. Once obtained, $\bm W_{\mathrm E}^*$ can be decomposed as $\bm W_{\mathrm E}^* = \bm v_0^*\left( \bm v_0^*\right) ^H$ via the Cholesky decomposition to recover the desired energy beamforming vector. 
Thus, we can choose to send only one energy beam to simplify the transmitter implementation by setting $\bm v_k = \bm v_0^*$ for any $k\in\mathcal{K_E}$ and $\bm v_j = \mathbf 0$, $\forall j\in\mathcal{K_E}\backslash\{k\}$.

\subsubsection{Optimizing $\mathbf \Theta$} For any given $\bm v_0^*$, $\mathbf \Theta$ can be optimized by solving (P1-NoIUs) with only the constraint in \eqref{WPT_cons:amp}. To facilitate the solution design, we define $\bm u = \left[u_1, \cdots, u_N \right]^H$, $\bar{\bm u} = \left[\bm u; 1\right]$, $\mathbf G_j = \left[ \text{diag}\left( \bm g_{r,j}^H\right) \mathbf F; \bm g_{d,j}^H \right] $, $\mathbf Z_j = \text{diag}\left(\left[ \bm g_{r,j}^H, 0\right] \right)\text{diag}\left(\left[ \bm g_{r,j}; 0\right] \right)$, $\mathbf \Phi = \text{diag}\left(\left[ \mathbf F \bm v_0^*; 0\right] \right)\left(\text{diag}\left(\left[ \mathbf F \bm v_0^*; 0\right] \right) \right)^H$, and $\mathbf P = \text{diag}\left(\left[\mathbf 1_{N\times 1}; 0\right] \right)$. Then, we have $\bm g_{r,j}^H\mathbf \Theta\mathbf F + \bm g_{d,j}^H = \bar{\bm u}^H\mathbf G_j$, $\left\|\bm g_{r,j}^H\mathbf \Theta\right\|^2 = \bar{\bm u}^H \mathbf Z_j\bar{\bm u}$, $\left\|\mathbf \Theta\mathbf F\bm v_0^*\right\|^2 = \bar{\bm u}^H \mathbf \Phi\bar{\bm u}$, and $\left\|\mathbf \Theta\right\|_F^2 = \bar{\bm u}^H\mathbf P\bar{\bm u}$. As a result, the subproblem is equivalent to 
\begin{subequations}\label{WPT_sub2_equiv}
	\begin{eqnarray}
	&\underset{\bar{\boldsymbol u}}{\max}& \sum_{j\in\mathcal {K_E}} \alpha_j \left|\bar{\bm u}^H\mathbf G_j\bm v_0^* \right|^2  + \sum_{j\in\mathcal {K_E}} \alpha_j \sigma_z^2\bar{\bm u}^H \mathbf Z_j\bar{\bm u} \label{WPT_sub2_equiv_obj}\\
	&\text{s.t.}& \bar{\bm u}^H \mathbf \Phi\bar{\bm u} + \sigma_z^2\bar{\bm u}^H\mathbf P\bar{\bm u} \leq P_{\text{I}}, \label{WPT_sub2_equiv_cons1}\\
	&& [\bar{\bm u}]_{N+1} = 1. \label{WPT_sub2_equiv_cons2}
	\end{eqnarray}
\end{subequations}
Although maximizing a convex function results in a non-convex problem \cite{2004_S.Boyd_cvx}, the convexity of the objective function in \eqref{WPT_sub2_equiv_obj} allows us to apply the iterative SCA technique for solving problem \eqref{WPT_sub2_equiv} suboptimally. To begin with, the objective function can be written in a compact form as $\bar{\bm u}^H\mathbf A\bar{\bm u}$, where $\mathbf A = \sum_{j\in\mathcal {K_E}} \alpha_j\left( \mathbf G_j\bm v_0^*(\bm v_0^*)^H\mathbf G_j^H + \sigma_z^2\mathbf Z_j\right) \geq 0$. Next, for a given local feasible point $\bar{\bm u}^{(l)}$ in the $l$-th iteration, the first-order Taylor expansion is a lower bound for $\bar{\bm u}^H\mathbf A\bar{\bm u}$ that can be expressed as
\begin{align}\label{WPT_sub2_equiv_obj_sca}
\bar{\bm u}^H\mathbf A\bar{\bm u} \geq 2{\rm Re}\left\lbrace{\bar{\bm u}^H\mathbf A\bar{\bm u}^{(l)}}\right\rbrace - \left( \bar{\bm u}^{(l)}\right)^H\mathbf A\bar{\bm u}^{(l)} \triangleq \mathcal G^{(l)}(\bar{\bm u}). 
\end{align}
By replacing the objective function in \eqref{WPT_sub2_equiv_obj} with $\mathcal G^{(l)}(\bar{\bm u})$, a suboptimal solution to \eqref{WPT_sub2_equiv} can be obtained by solving the following convex quadratically constrained quadratic program (QCQP): 
\begin{align}\label{WPT_sub2_sca} 
\underset{\bar{\boldsymbol u}}{\max} \ \mathcal G^{(l)}(\bar{\bm u}) \hspace{4mm} \text{s.t.} \  \eqref{WPT_sub2_equiv_cons1}, \eqref{WPT_sub2_equiv_cons2}. 
\end{align}
\begin{rem}\label{rem2}
	It is observed that the value of the LHS of constraint \eqref{WPT_sub2_equiv_cons1} is independent of the phase value of each element in $\bar {\bm u}$. Therefore, the optimal phase of $u_n$ can be obtained in a closed-form expression given by $\theta_n^* = 0$ if $\left[\mathbf A\bar{\bm u}^{(l)} \right]_n = 0$ and $\theta_n^* = {\rm arg}\left(\left[\mathbf A\bar{\bm u}^{(l)} \right]_n\right)$ otherwise, $\forall n$. 
	Subsequently, only the real-number magnitudes of $\{u_n\}$ are needed to be optimized by solving problem \eqref{WPT_sub2_sca} via existing convex optimization solvers, e.g., CVX \cite{2004_S.Boyd_cvx}. This helps to reduce the simulation time. 
\end{rem}
\vspace{-3mm}
\subsubsection{Convergence and Complexity Analysis}\label{subsubsec_WPT_complexity} As the objective value of (P1-NoIUs) is non-decreasing over the iterations and also upper-bounded by a finite value, the proposed algorithm is guaranteed to converge. Besides, the main computational burden stems from solving the SDP in \eqref{WPT_sub1_SDR} and the QCQP in \eqref{WPT_sub2_sca}. Simply speaking, given a solution accuracy $\varepsilon$, problem \eqref{WPT_sub1_SDR} can be solved with a computational complexity\footnote{According to \cite{2010_Imre_SDR_complexity}, for an SDP problem with $m$ SDP constraints, each of which involves an $n$-dimensional positive semidefinite matrix, the computational complexity for solving it is given by $\mathcal O\left( \sqrt{n}\log\left( \frac{1}{\varepsilon}\right)\left(mn^3 + m^2n^2 + m^3 \right)\right)$. For the SDP in \eqref{WPT_sub1_SDR}, we have $m=2$ and $n=M$.} of $\mathcal O\left(M^{3.5}\log(\frac{1}{\varepsilon})\right)$, while the arithmetic cost\footnote{According to \cite{1994_Nesterov_QCQP_complexity}, the arithmetic cost for solving a QCQP problem with $m$ variables and $n$ quadratic constraints is given by $\mathcal O\left(\sqrt{m}\left(mn^2 + n^3\right)\ln\left( \frac{2mV}{\varepsilon}\right)\right)$. For the QCQP in \eqref{WPT_sub2_sca}, we have $m=N+1$ and $n=2$.} of solving problem \eqref{WPT_sub2_sca} is less than $\mathcal O\left(N^{1.5}\ln\left( \frac{2(N+1)V}{\varepsilon}\right) \right)$, where $V$ is a constant defined in \cite{1994_Nesterov_QCQP_complexity} and $V > \varepsilon$. Thus, the total complexity of the proposed algorithm is about 
\begin{align} 
\mathcal O\left[\Upsilon\left(M^{3.5}\log\left( \frac{1}{\varepsilon}\right) + N^{1.5}\ln\left( \frac{2(N+1)V}{\varepsilon}\right)  \right)\right],
\end{align}
with $\Upsilon$ denoting the number of iterations required for convergence. 

\subsection{Are Dedicated Energy Beams Necessary? }\label{subsec:P1_answer}
Here, the general case where at least one IU coexists with $K_{\rm E}$ EUs is studied. 
In addition to the variables $\bm C$, $\bar P_{\text I}$, and $\bm W_{\mathrm E}$ defined in the previous subsection, we define $\bm W_i = \bm w_i\bm w_i^H$, $\forall i \in \mathcal{K_I}$. Then, it follows that $\bm W_i \succeq \bm 0$ and ${\rm rank}\left(\bm W_i \right) \leq 1 $, $\forall i\in\mathcal{K_I}$. By dropping the rank constraints on $\bm W_{\mathrm E}$ and $\{\bm W_i\}$, the SDR reformulation of (P1) can be expressed as  
\begin{subequations}\label{P1-SDR1}
	\begin{eqnarray}
	\hspace{-1.5cm}\text{(P1-SDR1)}: &\underset{\substack{\left\lbrace \boldsymbol W_i\in\mathbb H^M\right\rbrace, \\\boldsymbol W_{\mathrm E}\in\mathbb H^M, \mathbf \Theta}}{\max}& \sum_{i\in\mathcal {K_I}} {\rm tr}\left(\bm S \bm W_i \right) +  {\rm tr}\left(\bm S \bm W_{\mathrm E} \right) + \sum_{j\in\mathcal {K_E}}\alpha_j\sigma_z^2\left\|\bm g_{r,j}^H\mathbf \Theta\right\|^2 \label{P1-SDR1_obj}\\
	&\text{s.t.}& \frac{{\rm tr}\left(\bm h_i\bm h_i^H\bm W_i \right) }{\gamma_i} - \sum_{k\in \mathcal {K_I}\backslash\{i\}}{\rm tr}\left(\bm h_i\bm h_i^H\bm W_k \right) - {\rm tr}\left(\bm h_i\bm h_i^H\bm W_{\mathrm E}\right) \nonumber \\
	&& - \bar \sigma_i^2 \geq 0, \forall i\in\mathcal {K_I}, \label{SINR_orig}\\
	&& \sum_{i\in\mathcal {K_I}}{\rm tr}\left(\bm W_i \right) + {\rm tr}\left(\bm W_{\mathrm E}\right) \leq P_{\text{A}},\\
	&& \sum_{i\in\mathcal {K_I}}{\rm tr}\left(\bm C\bm W_i \right) + {\rm tr}\left(\bm C\bm W_{\mathrm E}\right) \leq \bar P_{\text{I}}, \label{P1-SDR1_cons:amp}\\
	&& \bm W_i \succeq \mathbf 0, \forall i\in\mathcal {K_I}, \; \bm W_{\mathrm E} \succeq \mathbf 0, 
	\end{eqnarray}
\end{subequations}
\normalsize where $\bar \sigma_i^2 \triangleq \sigma_z^2\left\|\mathbf h_{r,i}^H\mathbf \Theta\right\|^2 + \sigma_i^2$, $\forall i\in\mathcal{K_I}$. If we remove constraint \eqref{P1-SDR1_cons:amp}, neglect the IRS-amplified noise power in \eqref{P1-SDR1_obj} and \eqref{SINR_orig}, and set $\beta_n = 1$, $\forall n\in\mathcal N$, (P1-SDR1) is reduced to the same problem as in \cite{2020_Qingqing_SWIPT_letter} for passive IRS-aided SWIPT systems. It was proved in \cite[Proposition 1]{2020_Qingqing_SWIPT_letter} that transmitting dedicated energy signals is not necessary. In particular, the proof of \cite[Proposition 1]{2020_Qingqing_SWIPT_letter} relies on the result in \cite[Appendix A]{2013_Jie_SWIPT_arxiv}, which states that the optimal $\{\bm W_i\}$ and $\bm W_{\mathrm E}$ should all lie in the subspace spanned by one vector in the case that the optimal dual variables associated with the SINR constraints are all equal to zero. However, in the presence of constraint \eqref{P1-SDR1_cons:amp}, the above result cannot be proved to hold for (P1-SDR1) by following the same derivation as in \cite[Appendix A]{2013_Jie_SWIPT_arxiv}. Hence, for our considered active IRS-aided SWIPT system, we need to re-examine Proposition 1 in \cite{2020_Qingqing_SWIPT_letter}. Fortunately, by exploiting the structure of (P1-SDR1), we have the following theorem. 
\vspace{-3mm}
\begin{theo}\label{theo1}
\emph{Assuming that (P1-SDR1) is feasible for $P_{\rm A} > 0$, $P_{\rm I} > 0$, and $\gamma_i > 0$, $\forall i\in\mathcal{K_I}$, then there always exists an optimal solution to (P1-SDR1), denoted by $\left\lbrace\{\bm W_i^*\}, \bm W_{\mathrm E}^*, \mathbf \Theta^* \right\rbrace $, satisfying $\bm W_{\mathrm E}^* = \mathbf 0$ and ${\rm rank}\left(\bm W_i^*\right) = 1$, $\forall i\in\mathcal{K_I}$.} 
\vspace{-3mm}
\end{theo}
\begin{proof}
	Please refer to Appendix \ref{Appen_A}.
\end{proof}	
\vspace{-3mm}
\begin{rem}\label{rem3}
	\vspace{-3mm}
	It is worth mentioning that (P1-SDR1) under a non-linear EH model \cite{2021_Shayan_SWIPT} still has an optimal solution $\{\{\bm W_i^*\}, \bm W_{\mathrm E}^*, \mathbf \Theta^*\} $ such that $\bm W_{\mathrm E}^* = \mathbf 0$. The detailed proof is similar to that in Appendix \ref{Appen_A} and is therefore omitted. However, whether ${\rm rank}\left(\bm W_i^*\right) = 1$, $\forall i\in\mathcal{K_I}$ holds when $\bm W_{\mathrm E}^* = \mathbf 0$ remains unknown and needs further investigation. 
\end{rem}
\vspace{-3mm}
\begin{rem}\label{rem4}
	\vspace{-3mm}
	If the IUs have the capability of cancelling the interference due to the energy signals, the term $-{\rm tr}\left(\bm h_i\bm h_i^H\bm W_{\mathrm E}\right)$ in all the constraints in \eqref{SINR_orig} should be removed. In this case, if $K_{\rm I} = 1$, the result presented in Theorem \ref{theo1} can be similarly proved to hold. However, if $K_{\rm I} > 1$, $\bm W_{\mathrm E}^* \neq \mathbf 0$ holds at least in some specific channel conditions. For example, if the effective channels of the IUs and the EUs satisfy \cite[Assumption 1]{2013_Jie_SWIPT_arxiv}, there is ${\rm rank}(\bm W_{\mathrm E}^*) \leq 1$ according to \cite[Proposition 3.2]{2013_Jie_SWIPT_arxiv}. 
\end{rem}
\vspace{-3mm}

Theorem \ref{theo1} extends the result in \cite{2020_Qingqing_SWIPT_letter} by showing that even with the additional amplification power constraint in \eqref{P1-SDR1_cons:amp} and the non-negligible IRS-amplified noise power in \eqref{P1-SDR1_obj} and \eqref{SINR_orig}, the SDR is tight for (P1) and transmitting dedicated energy beams is not needed for achieving the optimal value of (P1). The intuitive explanation of this result is that transmitting dedicated energy beams would increase the interference power at the IUs while consuming power at the AP and at the IRS, and thus this should be avoided. By applying Theorem \ref{theo1}, the AP precoding design can be greatly simplified (especially when $K_{\rm E}$ is large) and (P1) is reduced to 
\begin{subequations}\label{P1_simp}
	\begin{eqnarray}
	&\underset{\{ \boldsymbol w_i\}, \mathbf \Theta}{\max}& \sum_{j\in\mathcal {K_E}} \alpha_j  \sum_{i\in\mathcal {K_I}}\left|\left(\bm g_{r,j}^H\mathbf \Theta\mathbf F + \bm g_{d,j}^H\right) \bm w_i \right|^2  + \sum_{j\in\mathcal {K_E}}\alpha_j\sigma_z^2\left\|\bm g_{r,j}^H\mathbf \Theta\right\|^2  \\
	&\text{s.t.}& \frac{\left|\left(\bm h_{r,i}^H\mathbf \Theta\mathbf F + \bm h_{d,i}^H\right)  \bm w_i \right|^2 }{\sum_{k\in \mathcal {K_I}\backslash\{i\}}\left|\left(\bm h_{r,i}^H\mathbf \Theta\mathbf F + \bm h_{d,i}^H\right) \bm w_k \right|^2 + \sigma_z^2\left\|\bm h_{r,i}^H\mathbf \Theta\right\|^2 + \sigma_i^2} \geq \gamma_i, \forall i\in\mathcal {K_I}, \\
	&&  \sum_{i\in\mathcal {K_I}}\left\|\bm w_i \right\|^2 \leq P_{\text{A}},\\
	&& \sum_{i\in\mathcal {K_I}}\left\|\mathbf \Theta\mathbf F\bm w_i\right\|^2 + \sigma_z^2\left\|\mathbf \Theta\right\|_F^2 \leq P_{\text{I}}. 
	\end{eqnarray}
\end{subequations}
\normalsize Although the problem at hand is simplified, it is still non-convex and difficult to solve, which motivates the development of the following algorithm.

\subsection{Proposed Algorithm for Problem \eqref{P1_simp}}\label{subsec:P1_alg}
Before solving problem \eqref{P1_simp}, we first transform it into an equivalent but more tractable form. Similar to Section \ref{subsec_WPT}, we define $\bm u = \left[u_1, \cdots, u_N \right]^H$,  $\bar{\bm u} = \left[\bm u; 1\right] $, $\mathbf G_j = \left[ \text{diag}\left( \bm g_{r,j}^H\right) \mathbf F; \bm g_{d,j}^H \right] $, $\mathbf H_i \!=\! \left[ \text{diag}\left( \bm h_{r,i}^H\right)\mathbf F; \bm h_{d,i}^H \right] $, $\mathbf Z_j \!=\! \text{diag}\left(\left[ \bm g_{r,j}^H, 0\right] \right)\text{diag}\left(\left[ \bm g_{r,j}; 0\right] \right)$, $\mathbf T_i \!=\! \text{diag}\left(\left[ \bm h_{r,i}^H, 0\right]\right)\text{diag}\left(\left[ \bm h_{r,i}; 0\right] \right)$, $\mathbf Q_i = \left[\left(\mathbf F\bm w_i\bm w_i^H\mathbf F^H \right) \odot \mathbf I_N, \mathbf 0; \mathbf 0\right] \in \mathbb C^{(N+1)\times(N+1)} $, and $\mathbf P = \text{diag}\left(\left[\mathbf 1_{N\times 1}; 0\right] \right)$. Then, we have $\bm g_{r,j}^H\mathbf \Theta\mathbf F + \bm g_{d,j}^H = \bar{\bm u}^H\mathbf G_j$, $\bm h_{r,i}^H\mathbf \Theta\mathbf F + \bm h_{d,i}^H = \bar{\bm u}^H\mathbf H_i$, $\left\|\bm g_{r,j}^H\mathbf \Theta\right\|^2 = \bar{\bm u}^H \mathbf Z_j\bar{\bm u}$, $\left\|\bm h_{r,i}^H\mathbf \Theta\right\|^2  = \bar{\bm u}^H \mathbf T_i\bar{\bm u}$, $\left\|\mathbf \Theta\mathbf F\bm w_i\right\|^2  = \bm u^H\left(\left(\mathbf F\bm w_i\bm w_i^H\mathbf F^H \right) \odot \mathbf I_N \right) \bm u = \bar{\bm u}^H\mathbf Q_i\bar{\bm u} $, and $\left\|\mathbf \Theta\right\|_F^2 = \bar{\bm u}^H\mathbf P\bar{\bm u}$. Therefore, problem \eqref{P1_simp} can be equivalently converted to
\begin{subequations}\label{P1_simp_eqv}
	\setlength\abovedisplayskip{6pt}
	\setlength\belowdisplayskip{6pt}
	\begin{eqnarray}
	&\underset{\{ \boldsymbol w_i\}, \bar{\boldsymbol u}}{\max}& \sum_{j\in\mathcal {K_E}} \alpha_j  \sum_{i\in\mathcal {K_I}}\left|\bar{\bm u}^H\mathbf G_j\bm w_i \right|^2  + \sum_{j\in\mathcal {K_E}}\alpha_j\sigma_z^2\bar{\bm u}^H \mathbf Z_j\bar{\bm u}  \\
	&\text{s.t.}& \left|\bar{\bm u}^H\mathbf H_i\bm w_i \right|^2 \geq \gamma_i\left(\sum_{k\in \mathcal {K_I}\backslash\{i\}}\left|\bar{\bm u}^H\mathbf H_i\bm w_k \right|^2 + \sigma_z^2\bar{\bm u}^H \mathbf T_i\bar{\bm u} + \sigma_i^2\right), \forall i\in\mathcal {K_I}, \\
	&& \sum_{i\in\mathcal {K_I}}\left\|\bm w_i \right\|^2 \leq P_{\text{A}},\\
	&& \sum_{i\in\mathcal {K_I}}\bar{\bm u}^H\mathbf Q_i\bar{\bm u} + \sigma_z^2\bar{\bm u}^H\mathbf P\bar{\bm u} \leq P_{\text{I}}, \\
	&& [\bar{\bm u}]_{N+1} = 1. 
	\end{eqnarray}
\end{subequations}
\normalsize Besides the matrices $\{\bm W_i\}$ defined in the previous subsection, we define another matrix $\bm U$, which is positive semidefinite and satisfies ${\rm rank}(\bm U) \leq 1$. By utilizing the cyclic property of the trace operator and dropping the rank constraints on $\{\bm W_i\}$ and $\bm U$, the SDR reformulation of problem \eqref{P1_simp_eqv} is given by
\begin{subequations}\label{P1_simp_SDR}
	\setlength\abovedisplayskip{6pt}
	\setlength\belowdisplayskip{6pt}
	\begin{align}
	&\hspace{-7mm}\underset{\substack{\{\boldsymbol W_i\in\mathbb H^M\},\\\boldsymbol U\in\mathbb H^{N+1}}}{\max} \  \sum_{j\in\mathcal {K_E}} \alpha_j \sum_{i\in\mathcal {K_I}}{\rm tr}\left(\mathbf G_j\bm W_i\mathbf G_j^H\bm U \right) + \sum_{j\in\mathcal {K_E}} \alpha_j \sigma_z^2{\rm tr}\left( \mathbf Z_j \bm U\right) \\
	\text{s.t.}\ & {\rm tr}\left(\mathbf H_i\bm W_i\mathbf H_i^H\bm U \right) \geq \gamma_i \left( \sum_{k\in \mathcal {K_I}\backslash\{i\}}{\rm tr}\left(\mathbf H_i\bm W_k\mathbf H_i^H\bm U \right) + \sigma_z^2{\rm tr}\left(\mathbf T_i\bm U \right) + \sigma_i^2\right),  \forall i\in\mathcal {K_I}, \label{P1_simp_SDR_cons:SINR}\\
	& \sum_{i\in\mathcal {K_I}}{\rm tr}\left(\bm W_i\right)  \leq P_{\text{A}},\label{P1_simp_SDR_cons:trans}\\
	& \sum_{i\in\mathcal {K_I}}{\rm tr}\left(\mathbf Q_i\bm U \right) + \sigma_z^2{\rm tr}\left(\mathbf P\bm U \right) \leq P_{\text I}, \\
	& [\bm U]_{N+1,N+1} = 1, \\
	& \bm W_i \succeq \mathbf 0, \forall  i\in\mathcal {K_I}, \ \bm U \succeq \mathbf 0.
	\end{align}
\end{subequations}
\normalsize However, problem \eqref{P1_simp_SDR} is still non-convex. Nevertheless, we note that by fixing either $\{\bm W_i\}$ or $\bm U$, problem \eqref{P1_simp_SDR} is reduced to a standard convex SDP that can be optimally solved by existing convex optimization solvers, e.g., CVX \cite{2004_S.Boyd_cvx}. Thus, this motivates us to utilize the AO method as in Section \ref{subsec_WPT} to solve problem \eqref{P1_simp_SDR} by iteratively optimizing $\{\bm W_i\}$ and $\bm U$ until convergence is achieved. \looseness=-1 

\emph{Convergence and Complexity Analysis:} The subproblem for updating  $\{\bm W_i\}$ or $\bm U$ is optimally solved in each iteration, and thus the objective value of problem \eqref{P1_simp_SDR} is non-decreasing over the iterations. This, together with the fact that the optimal value of problem \eqref{P1_simp_SDR} is bounded from above, guarantees the convergence of the proposed algorithm. After the convergence of the AO algorithm, if the solution is not rank-one, we can obtain a rank-one solution for $\{\bm W_i\}$ by applying the rank reduction techniques in \cite{2010_Yongwei_Rank}, and we can construct a rank-one solution for $\bm U$ by utilizing the Gaussian randomization method as in \cite{2018_Qingqing_Joint_conf}. It is shown by simulation results in Section \ref{Section_simulation} that the objective value of problem \eqref{P1_simp_SDR} achieved by the solution constructed utilizing the Gaussian randomization method is almost the same as that when the AO algorithm converges. Regarding the computational complexity of the proposed algorithm, it is dominated by solving the two SDP subproblems. Similar to the analysis in Section \ref{subsec_WPT}, the total computational complexity of the proposed algorithm is of the order of \cite{2010_Imre_SDR_complexity}
\begin{align}
\mathcal O\left[\mathcal L\log\left(\frac{1}{\varepsilon}\right) \big( \left(M^{3.5}+N^{3.5} \right)K_{\rm I} + \left(M^{2.5}+N^{2.5} \right)K_{\rm I}^2 + \left(M^{0.5}+N^{0.5} \right)K_{\rm I}^3 \big)\right],
\end{align}
\normalsize where $\varepsilon$ is the solution accuracy and $\mathcal L$ denotes the number of iterations needed for convergence.  

\vspace{-2.5mm}
\section{Proposed Solution to Problem (P2)}\label{Section_solution to (P2)}
In this section, we aim to solve (P2). First, we explore whether dedicated energy beams can be reasonably removed or set to zero when solving (P2). To facilitate the analysis, we start by introducing slack variables $\{\rho_i,\tau_i\}$, $i\in\mathcal{K_I}$ such that 
\begin{align}
e^{\rho_i} & = \sum_{k\in \mathcal {K_I}}\left|\bm h_i^H \bm w_k \right|^2 + \sum_{j\in \mathcal {K_E}}\left|\bm h_i^H \bm v_j \right|^2 + \bar \sigma_i^2, \forall i\in\mathcal{K_I}, \label{slack_equ_1}\\
e^{\tau_i} & = \sum_{k\in \mathcal {K_I}\backslash\{i\}}\left|\bm h_i^H \bm w_k \right|^2 + \sum_{j\in \mathcal {K_E}}\left|\bm h_i^H \bm v_j \right|^2 + \bar \sigma_i^2, \forall i\in\mathcal{K_I}, \label{slack_equ_2} 
\end{align}
where $\bar \sigma_i^2 \triangleq \sigma_z^2\left\|\mathbf h_{r,i}^H\mathbf \Theta\right\|^2 + \sigma_i^2$, $\forall i\in\mathcal{K_I}$. Then, the objective function of (P2) can be equivalently written as $\sum_{i\in\mathcal {K_I}}\mu_i \log_2e^{\left(\rho_i - \tau_i\right)}$. Consequently, (P2) can be reformulated as follows
\begin{subequations}\label{P2_Eqv} 
	\setlength\abovedisplayskip{5pt}
	\setlength\belowdisplayskip{5pt}
	\begin{eqnarray}
	\hspace{-1.5cm}\text{(P2-Eqv)}: &\underset{\substack{\left\lbrace \boldsymbol w_i\right\rbrace, \left\lbrace \boldsymbol v_j\right\rbrace, \\\mathbf \Theta, \left\lbrace \rho_i, \tau_i\right\rbrace}}{\max}& \sum_{i\in\mathcal {K_I}}\mu_i\log_2e^{ \left(\rho_i - \tau_i\right)}  \\	
	&\text{s.t.}& \sum_{k\in \mathcal {K_I}}\left|\bm h_i^H \bm w_k \right|^2 + \sum_{j\in \mathcal {K_E}}\left|\bm h_i^H \bm v_j \right|^2 + \bar \sigma_i^2 \geq e^{\rho_i}, \forall i\in\mathcal{K_I}, \label{P2_Eqv_cons:slack1}\\
	&& \sum_{k\in \mathcal {K_I}\backslash\{i\}}\left|\bm h_i^H \bm w_k \right|^2 + \sum_{j\in \mathcal {K_E}}\left|\bm h_i^H \bm v_j \right|^2 + \bar \sigma_i^2 \leq e^{\tau_i}, \forall i\in\mathcal{K_I}, \label{P2_Eqv_cons:slack2}\\
	&& \eqref{P2_cons:E}, \eqref{cons_PA}, \eqref{cons_PI}. 
	\end{eqnarray}
\end{subequations}
\normalsize Note that the constraints in \eqref{P2_Eqv_cons:slack1} and \eqref{P2_Eqv_cons:slack2} are obtained by replacing the equality signs in \eqref{slack_equ_1} and \eqref{slack_equ_2} with inequality signs. At the optimal solution to (P2-Eqv), the constraints in \eqref{P2_Eqv_cons:slack1} and \eqref{P2_Eqv_cons:slack2} must be
active, since otherwise the objective value can be further improved by increasing $\rho_i$ (decreasing $\tau_i$). Thus, (P2-Eqv) is equivalent to (P2). Recall that we defined $\bm W_i = \bm w_i\bm w_i^H$, $\forall i \in \mathcal{K_I}$, $\bm W_{\mathrm E} = \sum_{j\in \mathcal {K_E}}\bm v_j\bm v_j^H$, $\bm C = \mathbf F^H\bm \Theta^H \bm \Theta\mathbf F$, and $\bar P_{\text{I}} = P_{\text{I}} - \sigma_z^2\left\|\mathbf \Theta\right\|_F^2$ in the previous section. Then, the SDR reformulation of (P2-Eqv) can be expressed as
\begin{subequations}\label{P2_Eqv_SDR1}
	\begin{eqnarray}
    && \hspace{-3.3cm}\text{(P2-Eqv-SDR1)}: 
    \underset{\substack{\left\lbrace \boldsymbol W_i\in\mathbb H^M\right\rbrace, \\\boldsymbol W_{\mathrm E}\in\mathbb H^M, \mathbf \Theta, \left\lbrace \rho_i, \tau_i\right\rbrace}}{\max} \hspace{3mm} \sum_{i\in\mathcal {K_I}}\mu_i \log_2e^{\left(\rho_i - \tau_i\right)}  \\	
	&\text{s.t.}& \sum_{k\in\mathcal {K_I}}{\rm tr}\left(\bm h_i\bm h_i^H\bm W_k \right) + {\rm tr}\left(\bm h_i\bm h_i^H\bm W_{\mathrm E} \right) + \bar \sigma_i^2 \geq e^{\rho_i}, \forall i\in\mathcal{K_I}, \label{P2_Eqv_SDR1_cons:slack1}\\
	&& \sum_{k\in \mathcal {K_I}\backslash\{i\}}{\rm tr}\left(\bm h_i\bm h_i^H\bm W_k \right) + {\rm tr}\left(\bm h_i\bm h_i^H\bm W_{\mathrm E} \right) + \bar \sigma_i^2 \leq e^{\tau_i}, \forall i\in\mathcal{K_I}, \label{P2_Eqv_SDR1_cons:slack2}\\
    && \sum_{i\in\mathcal {K_I}}{\rm tr}\left(\bm g_j\bm g_j^H\bm W_i \right) + {\rm tr}\left(\bm g_j\bm g_j^H\bm W_{\mathrm E}\right) \geq \bar E_j, \forall j\in\mathcal{K_E}, \label{P2_Eqv_SDR1_cons:E}\\
    &&  \sum_{i\in\mathcal {K_I}}{\rm tr}\left(\bm W_i \right) + {\rm tr}\left(\bm W_{\mathrm E}\right) \leq P_{\text{A}}, \label{P2_Eqv_SDR1_cons:trans}\\
    &&  \sum_{i\in\mathcal {K_I}}{\rm tr}\left(\bm C\bm W_i \right) + {\rm tr}\left(\bm C\bm W_{\mathrm E}\right) \leq \bar P_{\text{I}}, \label{P2_Eqv_SDR1_cons:amp}\\
    &&  \bm W_i \succeq \bm 0, \forall i\in\mathcal {K_I}, \; \bm W_{\mathrm E} \succeq 0, \label{P2_Eqv_SDR1_cons:sedmi}
	\end{eqnarray}
\end{subequations}
\normalsize where $\bar E_j \triangleq E_j - \sigma_z^2\left\|\bm g_{r,j}^H\mathbf \Theta\right\|^2$, $\forall j\in\mathcal{K_E}$. Inspired by Theorem \ref{theo1}, we have the following theorem.
\begin{theo}\label{theo2}
	\vspace{-3mm}
	\emph{Assuming that (P2-Eqv-SDR1) is feasible for $P_{\rm A} > 0$, $P_{\rm I} > 0$, and $E_j\geq0$, $\forall j\in\mathcal{K_E}$, then it always has an optimal solution $\{\{\bm W_i^*\}, \bm W_{\mathrm E}^*, \mathbf \Theta^*, \{\rho_i^*, \tau_i^*\}\} $ such that $\bm W_{\mathrm E}^* = \mathbf 0$ and ${\rm rank}\left( \bm W_i^*\right) \leq 1, \forall i\in\mathcal{K_I'}$, where $\mathcal{K_I'} \subseteq \mathcal{K_I}$ and $\left|\mathcal{K_I'}\right| \geq K_{\rm I} - 1$. }
	\vspace{-3mm}
\end{theo}
\begin{proof}
	Please refer to Appendix \ref{Appen_B}.
\end{proof}	
\vspace{-3mm}
\vspace{-3mm}
\begin{rem}\label{rem5}
	For the case considering a non-linear EH model \cite{2021_Shayan_SWIPT}, the EH constraints in \eqref{P2_Eqv_SDR1_cons:E} can be replaced by
	\begin{align}
	\sum_{i\in\mathcal {K_I}}{\rm tr}\left(\bm g_j\bm g_j^H\bm W_i \right) + {\rm tr}\left(\bm g_j\bm g_j^H\bm W_{\mathrm E}\right) \geq P_j(E_j) - \sigma_z^2\left\|\bm g_{r,j}^H\mathbf \Theta\right\|^2, \forall j\in\mathcal{K_E},
	\end{align}
	where $P_j(E_j)$ is a function of $E_j$, whose value is known and fixed for a given $E_j$. It is not difficult to see that, by replacing the term $\bar E_j$ by $P_j(E_j) - \sigma_z^2\left\|\bm g_{r,j}^H\mathbf \Theta\right\|^2$ in Appendix \ref{Appen_B}, we can immediately obtain the same result presented in Theorem \ref{theo2} for the case adopting a non-linear EH model.
\end{rem}
\vspace{-3mm}
\begin{rem}\label{rem6}
	\vspace{-3mm}
	If the interference caused by the energy signals can be cancelled by the IUs, the term ${\rm tr}\left(\bm h_i\bm h_i^H\bm W_{\mathrm E}\right)$ in all the constraints in \eqref{P2_Eqv_SDR1_cons:slack1} and \eqref{P2_Eqv_SDR1_cons:slack2} should be removed. In this case, if $K_{\rm I} = 1$, the result revealed in Theorem \ref{theo2} can be similarly proved to hold. However, if $K_{\rm I} > 1$, the result may be different. Specifically, it is evident that (P2-Eqv-SDR1) without the term ${\rm tr}\left(\bm h_i\bm h_i^H\bm W_{\mathrm E}\right)$ always yields an equal or larger objective value than when it has the term ${\rm tr}\left(\bm h_i\bm h_i^H\bm W_{\mathrm E}\right)$. Meanwhile, (P2-Eqv-SDR1) with the term ${\rm tr}\left(\bm h_i\bm h_i^H\bm W_{\mathrm E}\right)$ does not need $\bm W_{\mathrm E}$ according to Theorem \ref{theo2}. Based on these facts, it can be concluded that a non-zero $\bm W_{\mathrm E}$ may help enhance the performance of  (P2-Eqv-SDR1) without the term ${\rm tr}\left(\bm h_i\bm h_i^H\bm W_{\mathrm E}\right)$. 
\end{rem}
\vspace{-3mm}
Since it is difficult to check if the SDR is tight for (P2-Eqv), it may not be concluded from Theorem \ref{theo2} that dedicated energy signals are not needed for achieving the optimality of (P2-Eqv). Nevertheless, Theorem \ref{theo2} indicates that the AP precoding design in (P2-Eqv-SDR1) can be simplified by setting $\bm W_{\mathrm E} = \mathbf 0$ without any loss of optimality. Additionally, if the obtained optimal solution to (P2-Eqv-SDR1) does not satisfy the rank-one constraints on $\{\bm W_i\}$, we can always construct an alternative optimal solution that fulfills the condition ${\rm rank}\left( \bm W_i\right) \leq 1$ for no less than $\left( K_{\rm I}-1\right) $ $i\in\mathcal{K_I}$, as presented in Appendix \ref{Appen_B}. These considerations motivate us to focus on solving a simplified version of (P2-Eqv-SDR1), denoted by (P2-Eqv-SDR2) that is obtained by setting $\bm W_{\mathrm E} = \mathbf 0$, instead of (P2-Eqv). To tackle (P2-Eqv-SDR2), we apply the AO method to decompose the problem into two subproblems which are then solved alternatingly until convergence is reached. The details are provided in the next subsections. 

\subsection{Transmit Beamforming Optimization}\label{P2_trans_opt}
For any given $\mathbf \Theta$, the subproblem for optimizing $\{\bm W_i\}$ can be written as 
\begin{subequations}\label{P2_Eqv_SDR1_sub1}
	\begin{eqnarray}
	&\underset{\left\lbrace \boldsymbol W_i\in\mathbb H^M\right\rbrace, \left\lbrace \rho_i, \tau_i\right\rbrace}{\max}& \sum_{i\in\mathcal {K_I}}\mu_i \log_2e^{\left(\rho_i - \tau_i\right)}  \label{P2_Eqv_SDR1_sub1_obj}\\	
	&\text{s.t.}& \sum_{k\in\mathcal {K_I}}{\rm tr}\left(\bm h_i\bm h_i^H\bm W_k \right) + \bar \sigma_i^2 \geq e^{\rho_i}, \forall i\in\mathcal{K_I}, \label{P2_Eqv_SDR1_sub1_cons:slack1}\\
	&& \sum_{k\in \mathcal {K_I}\backslash\{i\}}{\rm tr}\left(\bm h_i\bm h_i^H\bm W_k \right) + \bar \sigma_i^2 \leq e^{\tau_i}, \forall i\in\mathcal{K_I}, \label{P2_Eqv_SDR1_sub1_cons:slack2}\\
	&& \sum_{i\in\mathcal {K_I}}{\rm tr}\left(\bm g_j\bm g_j^H\bm W_i \right) \geq \bar E_j, \forall j\in\mathcal{K_E}, \label{P2_Eqv_SDR1_sub1_cons:E}\\
	&&  \sum_{i\in\mathcal {K_I}}{\rm tr}\left(\bm W_i \right) \leq P_{\text{A}}, \label{P2_Eqv_SDR1_sub1_cons:trans}\\
	&&  \sum_{i\in\mathcal {K_I}}{\rm tr}\left(\bm C\bm W_i \right) \leq \bar P_{\text{I}},\label{P2_Eqv_SDR1_sub1_cons:amp}\\
	&&  \bm W_i \succeq \bm 0, \forall i\in\mathcal {K_I}. \label{P2_Eqv_SDR1_sub1_cons:posi}
	\end{eqnarray}
\end{subequations}
\normalsize However, constraint \eqref{P2_Eqv_SDR1_sub1_cons:slack2} is non-convex since the right-hand-side (RHS) is convex with respect to $\tau_i$, leading to the non-convexity of problem \eqref{P2_Eqv_SDR1_sub1}. To tackle this issue, the SCA technique is leveraged as in Section \ref{subsec_WPT}. Specifically, we can replace the convex term $e^{\tau_i}$ in \eqref{P2_Eqv_SDR1_sub1_cons:slack2} with its first-order Taylor expansion at a given local feasible point $\tau_i^{(t)}$. Then, the LHS of constraint \eqref{P2_Eqv_SDR1_sub1_cons:slack2} is upper-bounded by
\begin{align}
\sum_{k\in \mathcal {K_I}\backslash\{i\}}{\rm tr}\left(\bm h_i\bm h_i^H\bm W_k \right) + \bar \sigma_i^2 \leq e^{\tau_i^{(t)}}\left(\tau_i - \tau_i^{(t)} + 1 \right), \forall i\in\mathcal{K_I}. \label{P2_Eqv_SDR1_cons:slack2_sca}
\end{align}
By replacing \eqref{P2_Eqv_SDR1_sub1_cons:slack2} with \eqref{P2_Eqv_SDR1_cons:slack2_sca}, a suboptimal solution to problem \eqref{P2_Eqv_SDR1_sub1} can be obtained by solving the following problem
\begin{subequations}\label{P2_Eqv_SDR1_sub1_SCA}
	\begin{eqnarray}
	&\underset{\left\lbrace \boldsymbol W_i\right\rbrace, \left\lbrace \rho_i, \tau_i\right\rbrace}{\max}& \sum_{i\in\mathcal {K_I}}\mu_i \log_2e^{\left(\rho_i - \tau_i\right)}  \\	
	&\text{s.t.}& \eqref{P2_Eqv_SDR1_sub1_cons:slack1}, \eqref{P2_Eqv_SDR1_cons:slack2_sca}, \eqref{P2_Eqv_SDR1_sub1_cons:E} - \eqref{P2_Eqv_SDR1_sub1_cons:posi}. 
	\end{eqnarray}
\end{subequations}
By direct inspection, problem \eqref{P2_Eqv_SDR1_sub1_SCA} is a convex SDP and hence it can be optimally solved by using existing convex optimization solvers, e.g., CVX \cite{2004_S.Boyd_cvx}. 

It is worth noting that, although problem \eqref{P2_Eqv_SDR1_sub1_SCA} belongs to the class of separable SDPs and has an optimal solution that satisfies the condition $\sum_{i\in\mathcal {K_I}}\left( {\rm rank}(\bm W_i)\right)^2 \leq 2K_{\rm I} + K_{\rm E} + 2$ according to \cite[Theorem 3.2]{2010_Yongwei_Rank}, this is not sufficient to prove that ${\rm rank}(\bm W_i) \leq 1$, $\forall i\in\mathcal{K_I}$ due to the arbitrariness of $K_{\rm I}$ and $K_{\rm E}$. In addition, by using the duality principle and the Karush-Kuhn-Tucker (KKT) conditions, the optimal solution to problem \eqref{P2_Eqv_SDR1_sub1_SCA} can be proved to satisfy ${\rm rank}(\bm W_i) \leq 1$, $\forall i\in\mathcal{K_I}$ provided that constraint \eqref{P2_Eqv_SDR1_sub1_cons:trans} is active at the optimal solution, the proof of which is similar to that of \cite[Theorem 1]{2020_Dongfang_rank1} and it is thus skipped for brevity. Unfortunately, constraint \eqref{P2_Eqv_SDR1_sub1_cons:trans} is not necessarily active due to the existence of constraint \eqref{P2_Eqv_SDR1_sub1_cons:amp}. 

\vspace{-6mm}
\subsection{Reflect Beamforming Optimization}\label{P2_ref_opt}
For any given $\{\bm W_i\}$, $\mathbf \Theta$ can be optimized by solving the following subproblem 
\begin{subequations}\label{P2_Eqv_SDR1_sub2}
	\begin{eqnarray}
	&\underset{\mathbf \Theta,\left\lbrace \rho_i, \tau_i\right\rbrace}{\max}& \sum_{i\in\mathcal {K_I}}\mu_i\log_2e^{\left(\rho_i - \tau_i\right)} \\
	&\text{s.t.}& \sum_{k\in\mathcal {K_I}}\bm h_i^H\bm W_k\bm h_i + \sigma_z^2\left\|\bm h_{r,i}^H\mathbf \Theta\right\|^2 + \sigma_i^2 \geq e^{\rho_i}, \forall i\in\mathcal{K_I}, \\
	&& \sum_{k\in \mathcal {K_I}\backslash\{i\}}\bm h_i^H\bm W_k\bm h_i + \sigma_z^2\left\|\bm h_{r,i}^H\mathbf \Theta\right\|^2 + \sigma_i^2 \leq e^{\tau_i}, \forall i\in\mathcal{K_I}, \\
	&& \sum_{i\in\mathcal {K_I}}\bm g_j^H\bm W_i\bm g_j + \sigma_z^2\left\|\bm g_{r,j}^H\mathbf \Theta\right\|^2 \geq E_j, \forall j\in\mathcal{K_E}, \\
	&& \sum_{i\in\mathcal {K_I}}{\rm tr}\left(\mathbf \Theta\mathbf F \bm W_i \mathbf F^H\mathbf \Theta^H \right)  + \sigma_z^2\left\|\mathbf \Theta\right\|_F^2 \leq P_{\text{I}}. 
	\end{eqnarray}
\end{subequations}
\normalsize To facilitate the solution of problem \eqref{P2_Eqv_SDR1_sub2}, we first transform it into a more tractable form. As in Section \ref{subsec:P1_alg}, we define $\bm u = \left[u_1, \cdots, u_N \right]^H$, $\bar{\bm u} = \left[\bm u; 1\right] $, $\mathbf H_i = \left[ \text{diag}\left( \bm h_{r,i}^H\right)\mathbf F; \bm h_{d,i}^H \right] $, $\mathbf G_j = \left[ \text{diag}\left( \bm g_{r,j}^H\right) \mathbf F; \bm g_{d,j}^H \right] $, $\mathbf T_i = \text{diag}\left(\left[ \bm h_{r,i}^H, 0\right]\right)\text{diag}\left(\left[ \bm h_{r,i}; 0\right] \right)$, $\mathbf Z_j = \text{diag}\left(\left[ \bm g_{r,j}^H, 0\right] \right)\text{diag}\left(\left[ \bm g_{r,j}; 0\right] \right)$, $\mathbf Q_i = \left[\left(\mathbf F\bm W_i\mathbf F^H \right) \odot \mathbf I_N, \mathbf 0; \mathbf 0\right] \in \mathbb C^{(N+1)\times(N+1)} $, and $\mathbf P = \text{diag}\left(\left[\mathbf 1_{N\times 1}; 0\right] \right)$. Then, we have $\bm h_i^H = \bar{\bm u}^H\mathbf H_i$, $\bm g_j^H = \bar{\bm u}^H\mathbf G_j$, $\left\|\bm h_{r,i}^H\mathbf \Theta\right\|^2  = \bar{\bm u}^H \mathbf T_i\bar{\bm u}$, $\left\|\bm g_{r,j}^H\mathbf \Theta\right\|^2 = \bar{\bm u}^H \mathbf Z_j\bar{\bm u}$, ${\rm tr}\left(\mathbf \Theta\mathbf F \bm W_i \mathbf F^H\mathbf \Theta^H \right)  = \bm u^H\left(\left(\mathbf F\bm W_i\mathbf F^H \right) \odot \mathbf I_N \right) \bm u = \bar{\bm u}^H\mathbf Q_i\bar{\bm u} $, and $\left\|\mathbf \Theta\right\|_F^2 = \bar{\bm u}^H\mathbf P\bar{\bm u}$. Accordingly, problem \eqref{P2_Eqv_SDR1_sub2} can be equivalently written as
\begin{subequations}\label{P2_Eqv_SDR1_sub2_eqv}
	\begin{eqnarray}
	&\underset{\bar{\boldsymbol u}, \left\lbrace \rho_i, \tau_i\right\rbrace}{\max}& \sum_{i\in\mathcal {K_I}}\mu_i \log_2e^{\left(\rho_i - \tau_i\right)}  \\	
	&\text{s.t.}& \sum_{k\in \mathcal {K_I}}\bar{\bm u}^H \mathbf H_i\bm W_k\mathbf H_i^H\bar{\bm u} + \sigma_z^2\bar{\bm u}^H \mathbf T_i\bar{\bm u} + \sigma_i^2 \geq e^{\rho_i}, \forall i\in\mathcal{K_I}, \label{P2_Eqv_SDR1_sub2_eqv_cons:slack1}\\
	&& \sum_{k\in \mathcal {K_I}\backslash\{i\}}\bar{\bm u}^H \mathbf H_i\bm W_k\mathbf H_i^H\bar{\bm u} + \sigma_z^2\bar{\bm u}^H \mathbf T_i\bar{\bm u} + \sigma_i^2 \leq e^{\tau_i}, \forall i\in\mathcal{K_I}, \label{P2_Eqv_SDR1_sub2_eqv_cons:slack2}\\
	&& \sum_{i\in \mathcal {K_I}}\bar{\bm u}^H\mathbf G_j\bm W_i\mathbf G_j^H\bar{\bm u} + \sigma_z^2\bar{\bm u}^H \mathbf Z_j\bar{\bm u} \geq E_j, \forall j\in\mathcal {K_E}, \label{P2_Eqv_SDR1_sub2_eqv_cons:E}\\
	&& \sum_{i\in\mathcal {K_I}}\bar{\bm u}^H\mathbf Q_i\bar{\bm u} + \sigma_z^2\bar{\bm u}^H\mathbf P\bar{\bm u} \leq P_{\text{I}}, \label{P2_Eqv_SDR1_sub2_eqv_cons:amp}\\
    && [\bar{\bm u}]_{N+1} = 1. \label{P2_Eqv_SDR1_sub2_eqv_cons:equa}
	\end{eqnarray}
\end{subequations}
Note that the quadratic terms in \eqref{P2_Eqv_SDR1_sub2_eqv_cons:slack1} and \eqref{P2_Eqv_SDR1_sub2_eqv_cons:E}, and the RHS of \eqref{P2_Eqv_SDR1_sub2_eqv_cons:slack2} are all convex functions, thus making \eqref{P2_Eqv_SDR1_sub2_eqv_cons:slack1}-\eqref{P2_Eqv_SDR1_sub2_eqv_cons:E} non-convex. To deal with these constraints, we employ the SCA technique. Specifically, given the local feasible points $\bar{\bm u}^{(t)}$ and $\tau_i^{(t)}$ at the $t$-th iteration, by replacing the convex terms mentioned above with their respective first-order Taylor expansion-based lower bounds, we can obtain a convex subset of \eqref{P2_Eqv_SDR1_sub2_eqv_cons:slack1}-\eqref{P2_Eqv_SDR1_sub2_eqv_cons:E}, as follows 
\begin{subequations}\label{P2_Eqv_SDR1_sub2_cons_sca}
	\begin{align}
	\sum_{k\in \mathcal {K_I}} \chi^{(t)}\big(\bar{\bm u}, \mathbf H_i\bm W_k\mathbf H_i^H\big)  + \sigma_z^2\chi\left(\bar{\bm u}, \mathbf T_i\right) + \sigma_i^2 & \geq e^{\rho_i}, \forall i\in\mathcal{K_I}, \label{P2_Eqv_SDR1_sub2_eqv_cons:slack1_sca}\\
	\sum_{k\in \mathcal {K_I}\backslash\{i\}}\bar{\bm u}^H \mathbf H_i\bm W_k\mathbf H_i^H\bar{\bm u} + \sigma_z^2\bar{\bm u}^H \mathbf T_i\bar{\bm u} + \sigma_i^2 & \leq e^{\tau_i^{(t)}}\left(\tau_i - \tau_i^{(t)} + 1 \right), \forall i\in\mathcal{K_I}, \label{P2_Eqv_SDR1_sub2_eqv_cons:slack2_sca}\\
	\sum_{i\in \mathcal {K_I}} \chi^{(t)}\big(\bar{\bm u}, \mathbf G_j\bm W_i\mathbf G_j^H\big) + \sigma_z^2\chi\left(\bar{\bm u}, \mathbf Z_j\right) & \geq E_j, \forall j\in\mathcal {K_E}, \label{P2_Eqv_SDR1_sub2_eqv_cons:E_sca}
	\end{align}
\end{subequations}
where $\chi^{(t)}(\bar{\bm u}, \mathbf B) \triangleq 2{\rm Re}\{\bar{\bm u}^H\mathbf B\bar{\bm u}^{(t)}\} - \left( \bar{\bm u}^{(t)}\right)^H\mathbf B\bar{\bm u}^{(t)}$, $\mathbf B \in \left\lbrace\mathbf H_i\bm W_k\mathbf H_i^H,\mathbf T_i, \mathbf G_j\bm W_i\mathbf G_j^H,\mathbf Z_j\right\rbrace $. As a result, a lower bound of the optimal solution to problem \eqref{P2_Eqv_SDR1_sub2_eqv} can be obtained by solving the following convex QCQP with off-the-shelf convex optimization solvers, e.g., CVX \cite{2004_S.Boyd_cvx}. 
\begin{subequations}\label{P2_Eqv_SDR1_sub2_eqv_sca}
	\begin{eqnarray}
	&\underset{\bar{\boldsymbol u}, \left\lbrace \rho_i, \tau_i\right\rbrace}{\max}& \sum_{i\in\mathcal {K_I}}\mu_i \log_2e^{\left(\rho_i - \tau_i\right)}  \\	
	&\text{s.t.}& \eqref{P2_Eqv_SDR1_sub2_cons_sca}, \eqref{P2_Eqv_SDR1_sub2_eqv_cons:amp}, \eqref{P2_Eqv_SDR1_sub2_eqv_cons:equa}. 
	\end{eqnarray}
\end{subequations}

\subsection{Overall Algorithm}
Based on Sections \ref{P2_trans_opt} and \ref{P2_ref_opt}, we propose an efficient algorithm for (P2-Eqv) by applying the AO method. Specifically, we solve (P2-Eqv-SDR2) by alternatingly solving problems \eqref{P2_Eqv_SDR1_sub1_SCA} and \eqref{P2_Eqv_SDR1_sub2_eqv_sca} until convergence is achieved, where the obtained solution in each iteration is used as the input of the next iteration. The solution once convergence is reached is denoted by $\grave{\Pi} \triangleq \{\{\grave{\bm W}_i\}, \grave{\mathbf \Theta}, \{\grave{\rho_i}, \grave{\tau_i}\}\} $. If ${\rm rank}\big(\grave{\bm W}_i\big) \leq 1$ holds for all $i\in\mathcal{K_I}$, the transmit precoder $\{\grave{\bm w}_i\}$ can be recovered from $\{\grave{\bm W}_i\}$ via the Cholesky decomposition. Otherwise, we define 
\begin{subequations}\label{trans_prec_recover}
	\begin{align}
	&\bar{\bm w}_i = \left(\bm h_i^H\grave{\bm W}_i\bm h_i\right)^{-1/2}\grave{\bm W}_i\bm h_i, \bar{\bm W}_i = \bar{\bm w}_i(\bar{\bm w}_i)^H, \forall i\in \mathcal {K_I}\backslash\{m\}, \\
	&\bar{\bm W}_m = \sum_{i\in\mathcal {K_I}}\grave{\bm W}_i - \sum_{i\in\mathcal {K_I}\backslash\{m\}}\bar{\bm W}_i,
	\end{align}
\end{subequations}
\normalsize where $m$ can be any element in the set $\mathcal {K_I}$. 
According to the proof of Theorem \ref{theo2} in Appendix \ref{Appen_B}, $\bar{\Pi} \triangleq \{\{\bar{\bm W}_i\}, \grave{\mathbf \Theta}, \{\grave{\rho_i}, \grave{\tau_i}\}\} $ is a feasible solution to (P2-Eqv-SDR2), and the objective values attained at $\grave{\Pi}$ and $\bar{\Pi}$ are the same. In this case, $\{\bar{\bm w}_i\}$ can be obtained by using \eqref{trans_prec_recover} and performing the Gaussian randomization method over $\bar{\bm W}_m$. 

Similar to the analyses in Sections \ref{subsec_WPT} and \ref{subsec:P1_alg}, the proposed algorithm is guaranteed to converge and the overall computational complexity is about \cite{2010_Imre_SDR_complexity,1994_Nesterov_QCQP_complexity}\begin{align}
O\left[\mathcal T\left(\sqrt{b}\left(ab^3 + a^2b^2 + a^3 \right)\log\left(\frac{1}{\varepsilon}\right) + \sqrt{c}\left(ca^2 + a^3 \right)\ln\left( \frac{2c V}{\varepsilon}\right) \right)\right],
\end{align}
where $a\triangleq 2K_{\rm I} + K_{\rm E} + 2$, $b \triangleq M$, $c\triangleq 2K_{\rm I} + N + 1$, $\varepsilon$ is the prescribed accuracy, $V$ is a constant defined in \cite{1994_Nesterov_QCQP_complexity}, and $\mathcal T$ denotes the number of iterations needed for convergence. 

\begin{figure}[!t]
	\centering
	\includegraphics[width=0.7\textwidth]{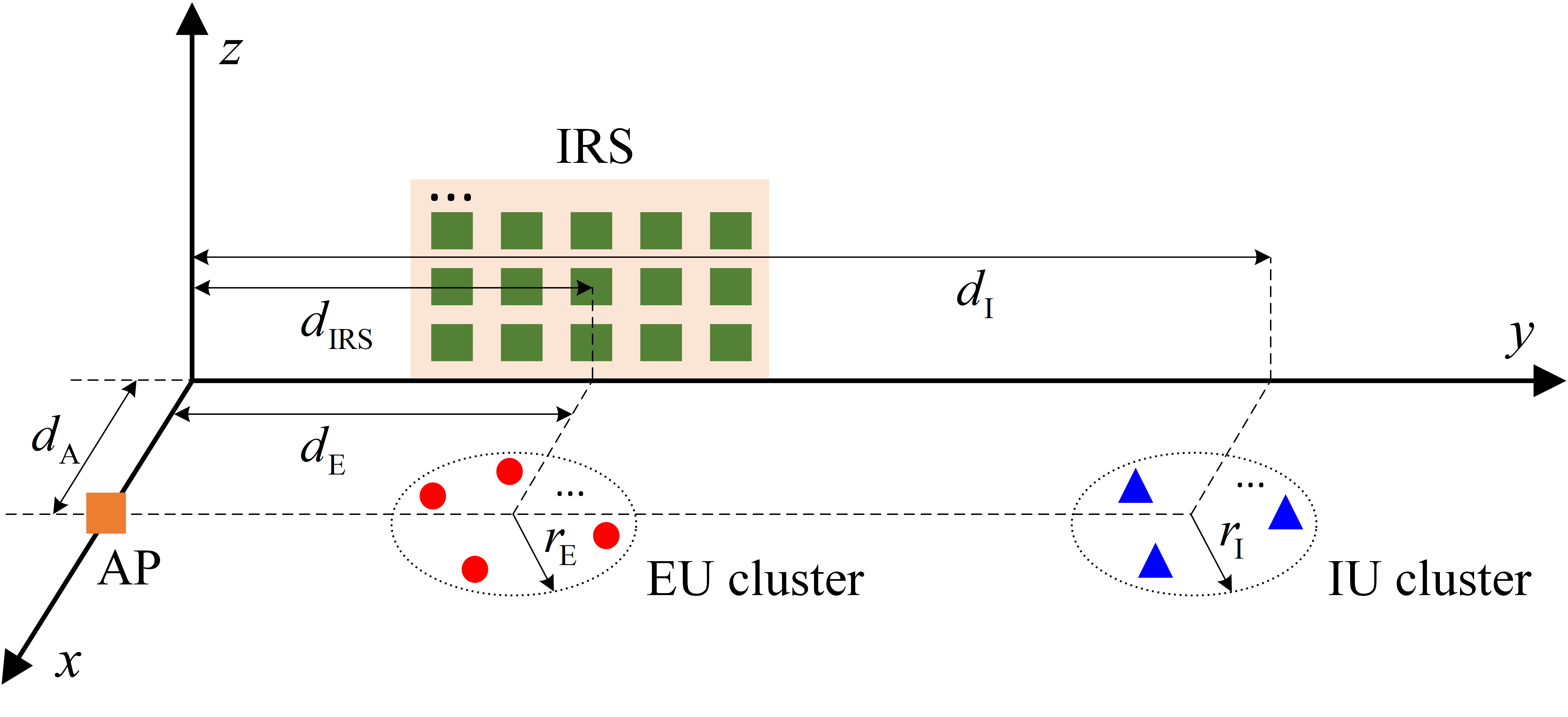}
	\caption{Simulation setup. } \label{Fig:simulation setup}
	\vspace{-5mm}
\end{figure}

\section{Simulation Results}\label{Section_simulation}
In this section, we demonstrate the effectiveness of the proposed algorithms with the aid of numerical simulations. As illustrated in Fig. \ref{Fig:simulation setup}, a three-dimensional (3D) coordinate setup is considered, where the AP and the IRS are located at $\left(d_{\rm A}, 0, 0 \right) $ and $\left(0, d_{\rm IRS}, 0 \right) $, respectively. The EUs and the IUs are randomly and uniformly distributed within two disks centered at $\left(d_{\rm A}, d_{\rm E}, 0 \right)$ and $\left(d_{\rm A}, d_{\rm I}, 0 \right)$ with radii equal to $r_{\rm E}$ and $r_{\rm I}$, respectively. The system is assumed to operate on a carrier frequency of $750$ MHz, with a wavelength $\lambda_c = 0.4$ meter (m) \cite{2020_Qingqing_SWIPT_QoS,2021_Changsheng_active_or_passive}. The large-scale path loss is modeled as $L(d) = C_0\left(d/D_0\right)^{(-\alpha)} $ \cite{2019_Qingqing_Joint}, where $C_0 = \left( \frac{\lambda_c}{4\pi}\right)^2$ is the path loss at the reference distance $D_0 = 1$ m, $d$ denotes the link distance, and $\alpha$ represents the path loss exponent. The path loss exponents of the AP-IRS, IRS-user, and AP-user links are set equal to $\alpha_{\rm AI} = 2.2$ \cite{2020_Qingqing_SWIPT_letter}, $\alpha_{\rm Iu} = 2.2$ \cite{2020_Qingqing_SWIPT_letter}, and $\alpha_{\rm Au} = 3.2$, respectively. We assume that the AP-IRS and the IRS-user links experience Rician fading with a Rician factor of $3$ dB, while the AP-user links undergo Rayleigh fading \cite{2022_zhendong_WPCN_robust}. In addition, we set $\alpha_j = 1$, $\forall j\in\mathcal{K_E}$ in (P1) and $\mu_i = 1$, $\forall i\in\mathcal{K_I}$ in (P2), i.e., the sum-power harvested by all the EUs and the sum-rate of all the IUs are considered, respectively \cite{2020_Qingqing_SWIPT_letter,2020_Cunhua_SWIPT}. 
Unless otherwise specified, other system parameters are set as follows: $\sigma_z^2 = \sigma_i^2 = -80$ dBm \cite{2021_Ruizhe_active_SIMO}, $\gamma_i = \gamma$, $\forall i\in\mathcal{K_I}$, $E_j = E$, $\forall j\in\mathcal{K_E}$, $N = 50$, $M = 5$, $d_{\rm A} = 3$ m, $d_{\rm I} = 100$ m, and $r_{\rm E} = r_{\rm I} = 2$ m.

For comparison purposes, we focus on the following two benchmark schemes: 1) Identical amplitudes: all active elements are assumed to have identical amplitudes, i.e., $\beta_n = \beta$, $\forall n\in\mathcal N$, and $\beta$ is optimized; 2) Passive IRS: we set $\beta_n = 1$, $\forall n\in\mathcal N$, neglect the noise introduced by the IRS, and remove the amplification power constraint. Moreover, for a fair comparison, we assume that the AP's total transmit power budget is $P_{\rm A} + P_{\rm I}$ in this benchmark scheme. The simulation results are obtained by averaging $100$ independent realizations of the channels and the users' locations. 

\begin{figure}[!t]
	\setlength{\abovecaptionskip}{-0.7mm}
	\begin{minipage}[t]{0.5\linewidth}
		\centering
		\includegraphics[scale=0.65]{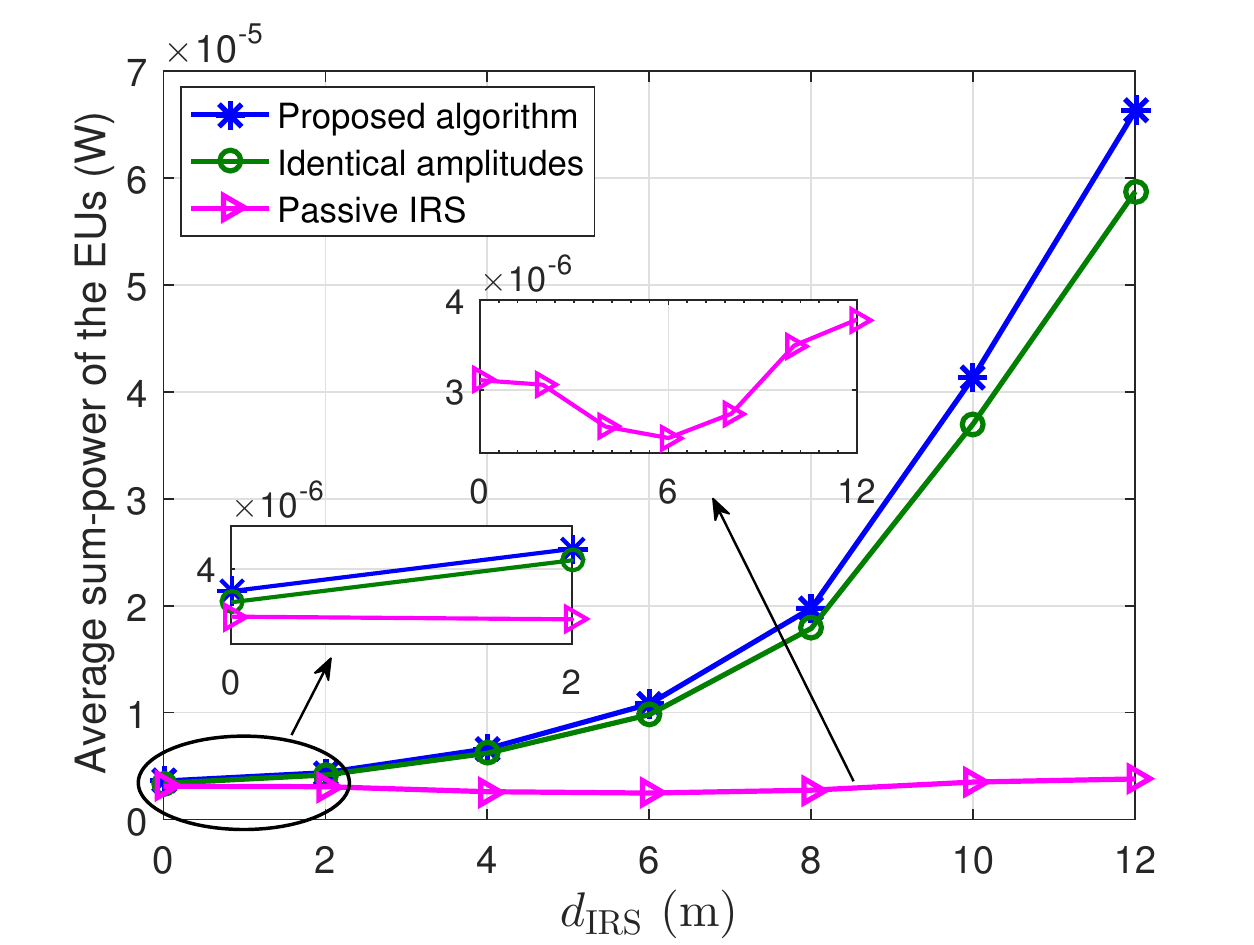}
		\caption{Average sum-power of the EUs versus the y-axis\\ coordinate value of the IRS.}
		\label{fig:WPT_WSP_vs_dIRS}
	\end{minipage}%
	\begin{minipage}[t]{0.5\linewidth}
		\centering
		\includegraphics[scale=0.65]{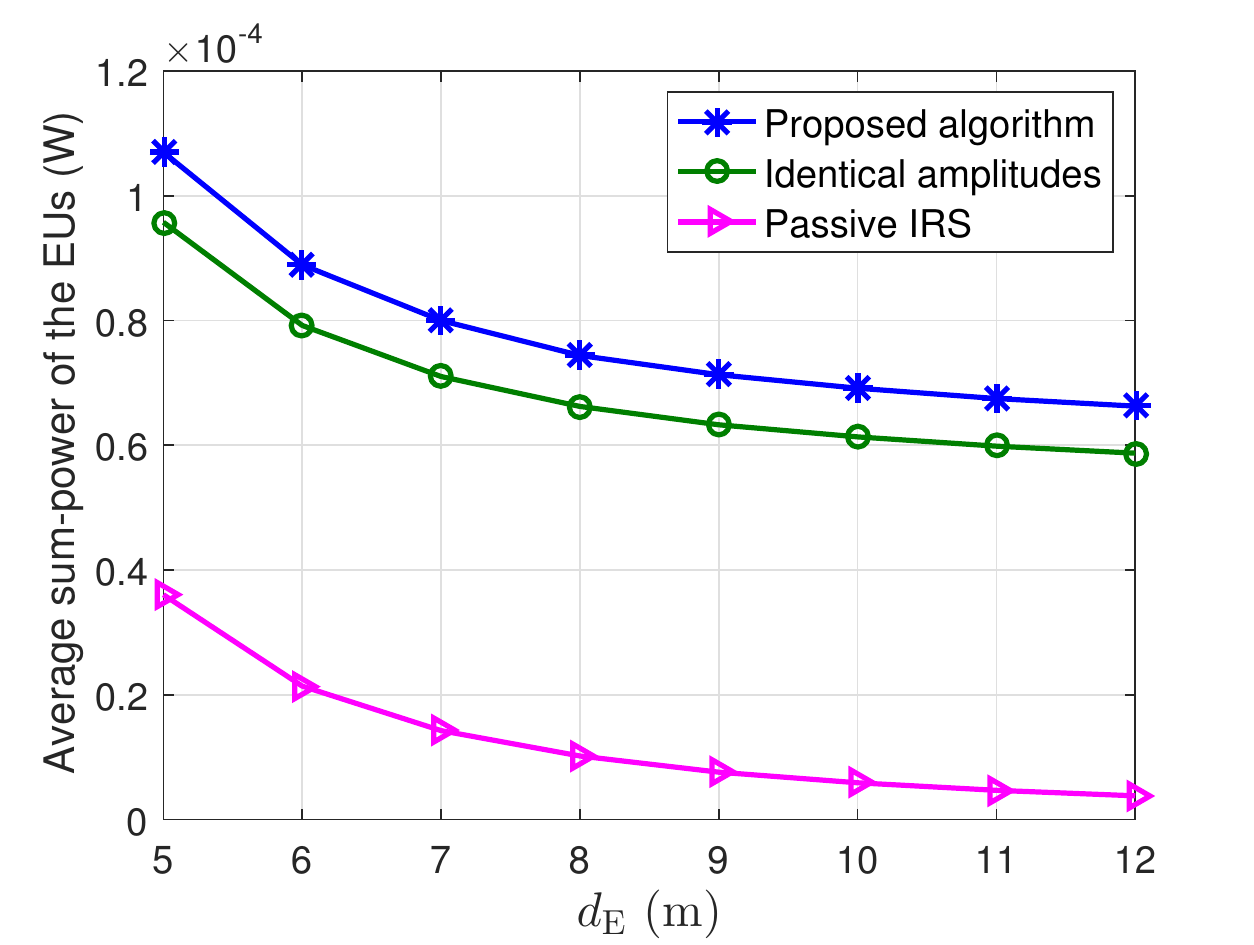}
		\caption{Average sum-power of the EUs versus the distance between the AP and the center of the disk of the EU cluster, where $d_{\rm IRS} = d_{\rm E}$. }
		\label{fig:WPT_WSP_vs_dE}
	\end{minipage}
	\vspace{-6mm}
\end{figure}

\subsection{Weighted Sum-power Maximization}
\subsubsection{Special Case with No IUs}
We first investigate a special case of the weighted sum-power maximization problem (P1) where there exist no IUs, i.e., $\mathcal{K_I} = \emptyset$. By varying the value of $d_{\rm IRS}$, we examine in Fig. \ref{fig:WPT_WSP_vs_dIRS} the average sum-power of $K_{\rm E} = 4$ EUs with $P_{\rm A} = 23$ dBm, $P_{\rm I} = 5$ dBm, and $d_{\rm E} = 12$ m. When deploying a passive IRS, as expected, it is observed that the EUs harvest the lowest sum-power when the IRS is located far from both the AP and the EU cluster (i.e., $d_{\rm IRS} = 6$ m) due to the product path loss attenuation law. When deploying an active IRS, in contrast, the sum-power of EUs increases drastically as the IRS moves closer to the EU cluster, since the incident signal power at the IRS becomes weaker with increasing $d_{\rm IRS}$, the active IRS can provide higher amplification gain according to constraint \eqref{WPT_cons:amp}, which compensates for the product path loss attenuation and contributes to an increase in the sum-power harvested by the EUs. This result indicates that to reap the full benefits of an active IRS, we should deploy it close to the EUs. 
Besides, it can been seen that the proposed algorithm outperforms the scheme with identical amplitudes at the active IRS as well as the scheme employing a passive IRS. The performance loss caused by adopting identical amplitudes for all elements shows the importance of properly designing the amplitudes at the active IRS for enhancing the system performance. 

To further demonstrate the benefits brought by the active IRS to WPT, we plot in Fig. \ref{fig:WPT_WSP_vs_dE} the average sum-power of the EUs versus the distance between the AP and the center of the disk of the EU cluster, where the IRS moves with the EU cluster to maintain the condition $d_{\rm IRS} = d_{\rm E}$. The other parameters are the same as those in Fig. \ref{fig:WPT_WSP_vs_dIRS}. As can be seen, by deploying an active IRS around the EUs, their sum harvested power is observably improved compared to the case with a passive IRS. This observation suggests that deploying an active IRS is more effective than deploying a passive IRS in extending the WPT operating range. 

\begin{figure}[!tbp]
	\setlength{\abovecaptionskip}{-0.2mm}
	\begin{minipage}[t]{0.5\linewidth}
	\centering
	\includegraphics[scale=0.65]{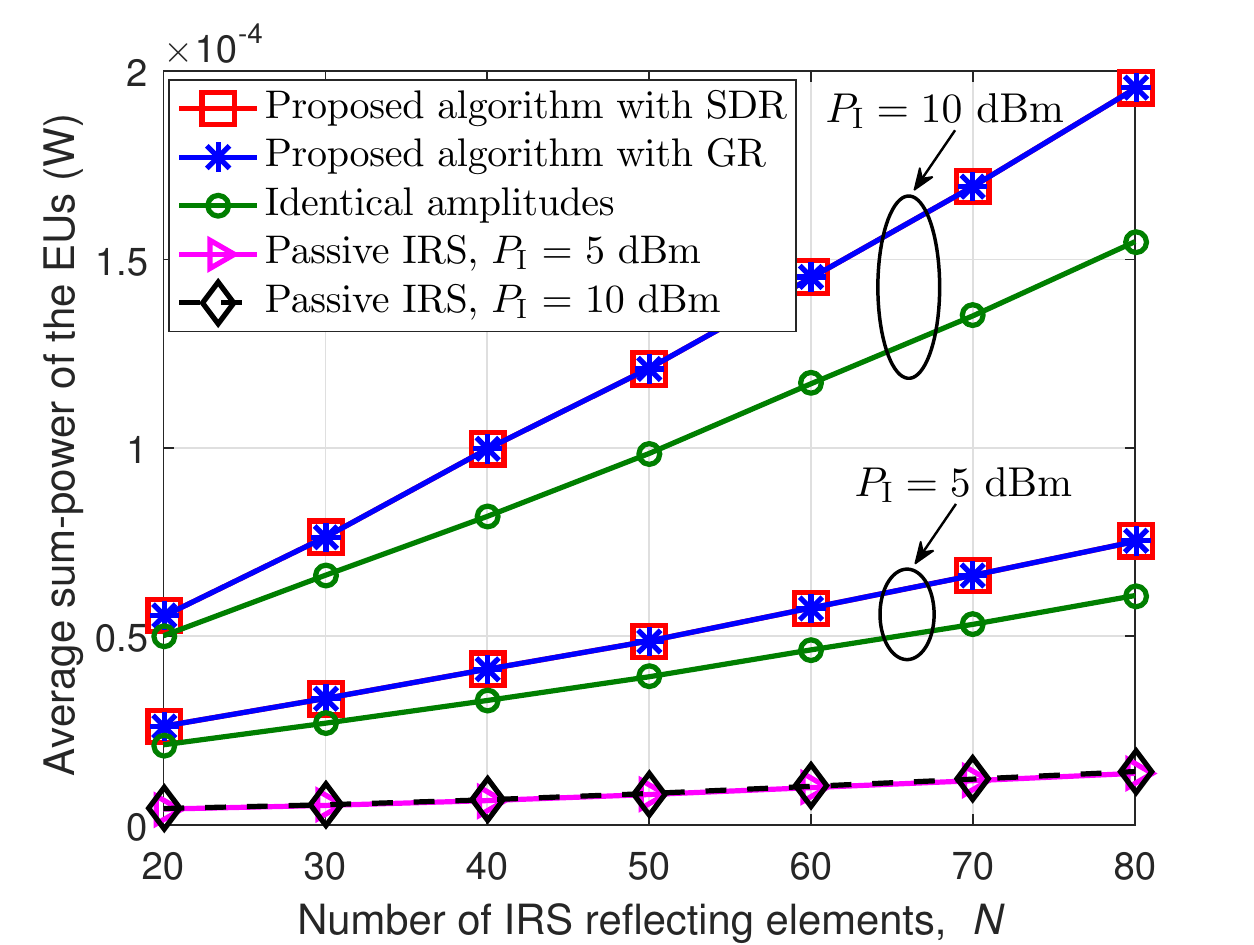}
	\caption{Average sum-power of the EUs versus the number\\ of IRS elements.}
	\label{fig:SWIPT_WSP_vs_N}
    \end{minipage}
	\begin{minipage}[t]{0.5\linewidth}
		\centering
		\includegraphics[scale=0.65]{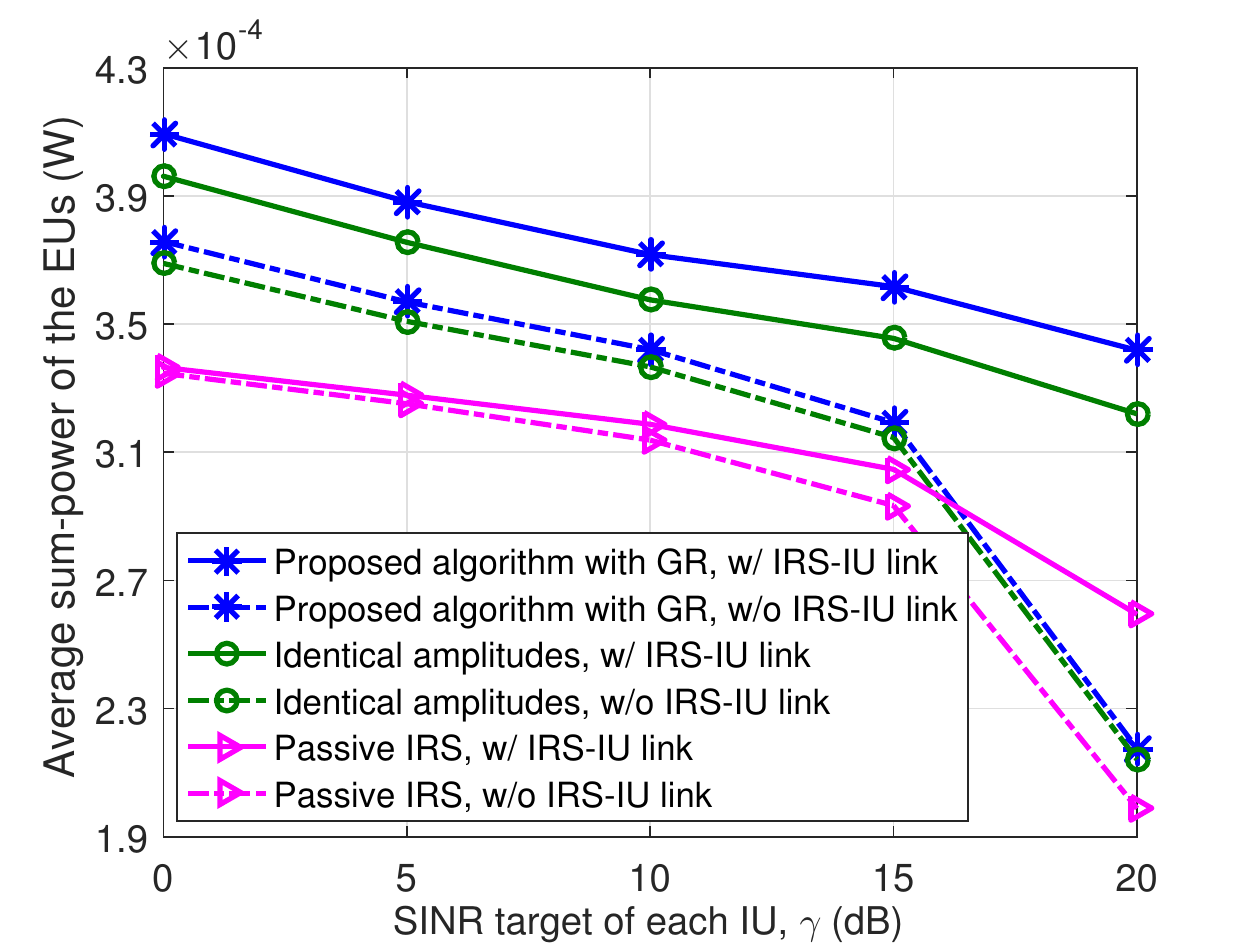}
		\caption{Average sum-power of the EUs versus the SINR target of each IU.}
		\label{fig:SWIPT_WSP_vs_SINR}
	\end{minipage}%
    \vspace{-5mm}
\end{figure}

\subsubsection{General Case with EUs and IUs Coexisting}
Next, we consider the general case of (P1) where both the EUs and the IUs exist. In Fig. \ref{fig:SWIPT_WSP_vs_N}, we show the average sum-power of the EUs versus the number of IRS elements with $K_{\rm E} = 4$, $K_{\rm I} = 2$, $P_{\rm A} = 23$ dBm, $\gamma = 5$ dB, and $d_{\rm IRS} = d_{\rm E} = 8$ m. Note that the ``Proposed algorithm with SDR'' corresponds to the solution when the proposed algorithm for the SDR problem in \eqref{P1_simp_SDR} converges, while the ``Proposed algorithm with GR'' utilizes the Gaussian randomization method to construct a rank-one $\bm U$ based on the solution obtained with the SDR method. It is observed that the performance of the proposed algorithm with GR closely approaches that achieved by the SDR method. One can also conclude from Fig. \ref{fig:SWIPT_WSP_vs_N} that the deployment of an active IRS is more suitable for application to space-limited scenarios, since it can significantly reduce the required surface size for achieving a given performance level. Furthermore, we observe that increasing $N$ widens the performance gap between the proposed algorithm and the other two benchmark schemes, as the former offers a better utilization of the system resources. Finally, we can see that a larger $P_{\rm I}$ leads to a better performance of the SWIPT system with an active IRS. This is excepted since an active IRS can provide a higher amplification gain for a larger value of $P_{\rm I}$. In contrast, the increase of $P_{\rm I}$ only brings a negligible performance gain to the passive IRS scheme as the value of $P_{\rm A}$, which is significantly larger than the two different values of $P_{\rm I}$, dominates the performance of this scheme. 

In Fig. \ref{fig:SWIPT_WSP_vs_SINR}, we study the average sum-power of the EUs versus the SINR target of each IU with $P_{\rm A} = 30$ dBm and $P_{\rm I} = 10$ dBm (other parameters are set to be the same as in Fig. \ref{fig:SWIPT_WSP_vs_N}). Two cases with and without the IRS-IU link are considered. As expected, the proposed algorithm significantly outperforms the other two benchmark schemes over the whole considered SINR regime. Another observation is that the performance gap between the two cases of the proposed algorithm is much larger than that of the passive IRS scheme, which shows the advantage of deploying an active IRS for effective WIT even when the IUs are far away from the IRS. 

\begin{figure}[!tbp]
	\setlength{\abovecaptionskip}{-0.2mm}
	\begin{minipage}[t]{0.5\linewidth}
		\centering
		\includegraphics[scale=0.64]{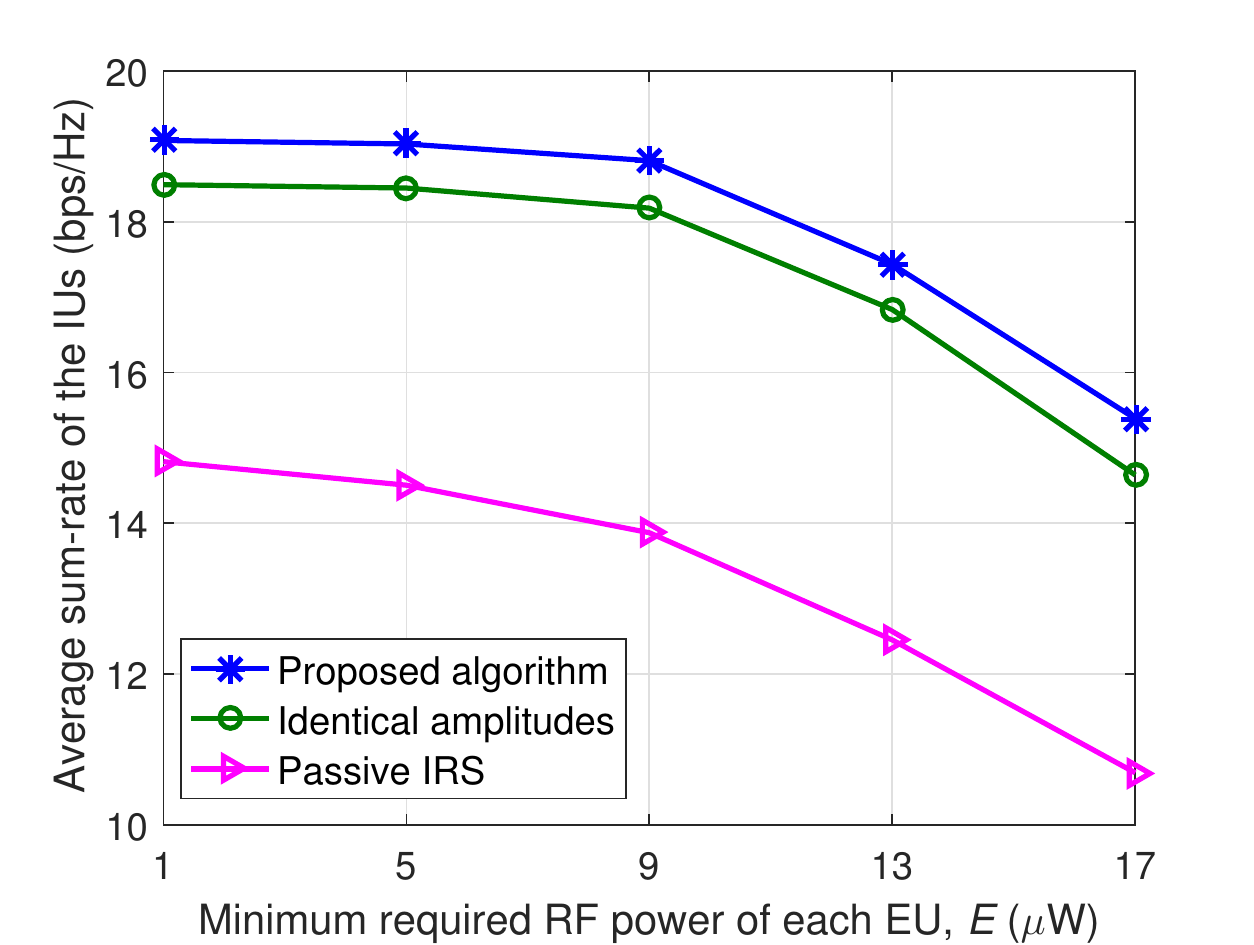}
		\caption{Average sum-rate of the IUs versus the minimum\\ required RF power of each EU.}
		\label{fig:WSR_vs_E}
	\end{minipage}%
	\begin{minipage}[t]{0.5\linewidth}
		\centering
		\includegraphics[scale=0.64]{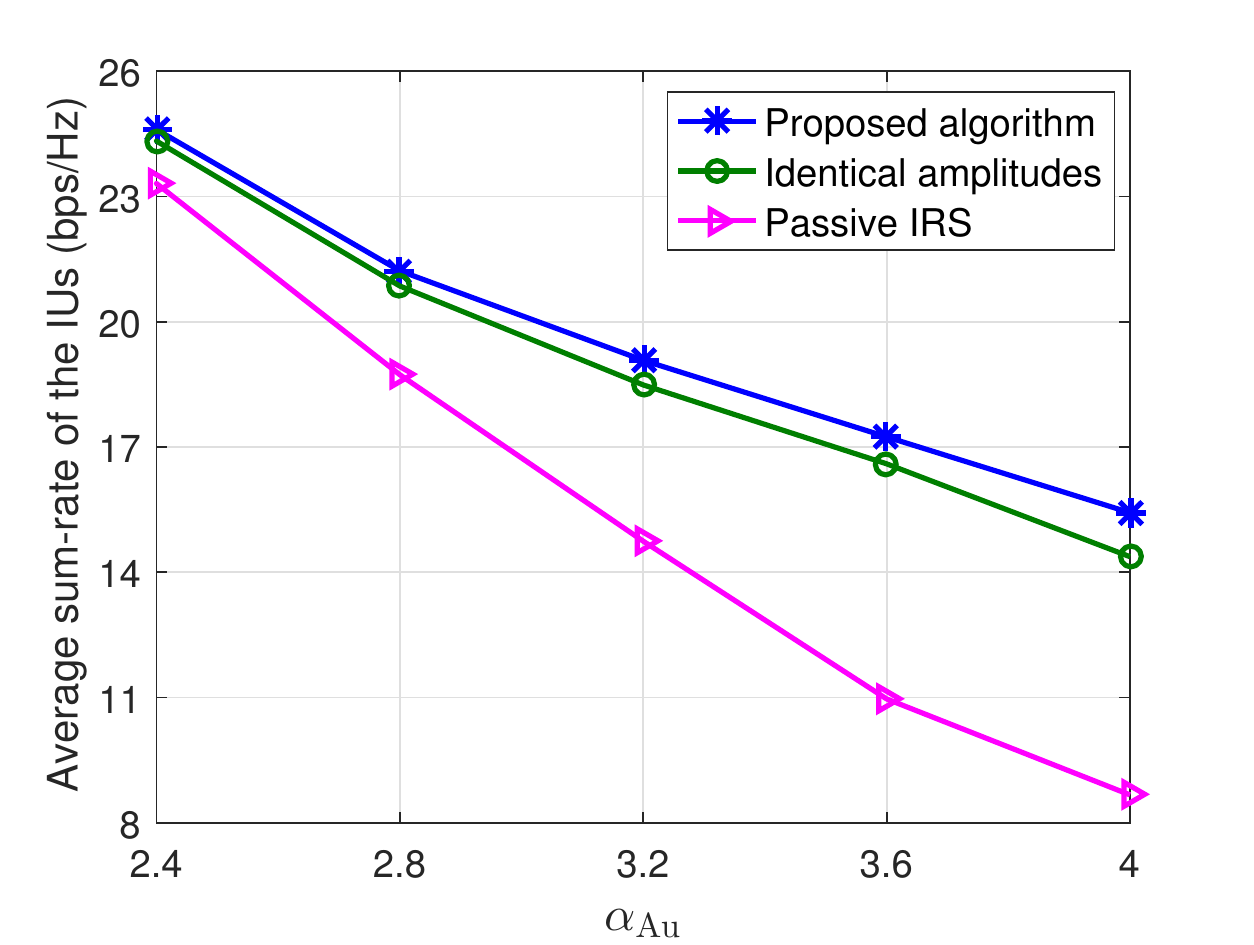}
		\caption{Average sum-rate of the IUs versus the path loss exponent of the AP-user link.}
		\label{fig:WSR_vs_pathloss}
	\end{minipage}
	\vspace{-4mm}
\end{figure}

\subsection{Weighted Sum-Rate Maximization}
In this subsection, we evaluate the performance of the proposed algorithm for the weighted sum-rate maximization problem (P2). Fig. \ref{fig:WSR_vs_E} illustrates the average sum-rate of IUs versus the minimum required RF power of each EU with $K_{\rm E} = K_{\rm I} = 2$, $P_{\rm A} = 30$ dBm, $P_{\rm I} = 10$ dBm, and $d_{\rm IRS} = d_{\rm E} = 8$ m. It is observed that the proposed algorithm performs much better than the passive IRS scheme, and the performance gap is larger with the increase of $E$. This again indicates that the deployment of an active IRS is beneficial to both WIT and WPT. 

Under the same setup as in Fig. \ref{fig:WSR_vs_E}, we study in Fig. \ref{fig:WSR_vs_pathloss} the impact of the path loss exponent of the AP-user direct link on the performance of different schemes when $E = 3$ $\mu$W. It can be seen that deploying an active IRS is more advantageous than deploying a passive IRS regardless of the value of $\alpha_{\rm Au}$. In addition, the performance gap between the proposed algorithm and the passive IRS scheme increases markedly with $\alpha_{\rm Au}$. This is expected since the performance of both schemes is dominated by the AP-user direct link when $\alpha_{\rm Au}$ is small while by the AP-IRS-user reflected link when $\alpha_{\rm Au}$ is large. 

In Fig. \ref{fig:WSR_vs_noise}, the impact of the noise variance of the active IRS on the performance of the proposed algorithm is investigated. Here, we set $E = 3$ $\mu$W and the other parameters are the same as those in Fig. \ref{fig:WSR_vs_E}. We can observe that even if $\sigma_z^2$ is very large, an active IRS can outperform a passive IRS via a proper joint optimization of the transmit and reflect beamforming. Moreover, the performance degradation of the proposed algorithm due to the increase of $\sigma_z^2$ is very small, sometimes even negligible when the IUs are far away from the IRS (i.e., $d_{\rm I} = 100$ m).

\begin{figure}[!tbp]
	\setlength{\abovecaptionskip}{-0.2mm}
	\begin{minipage}[t]{0.5\linewidth}
		\centering
		\includegraphics[scale=0.63]{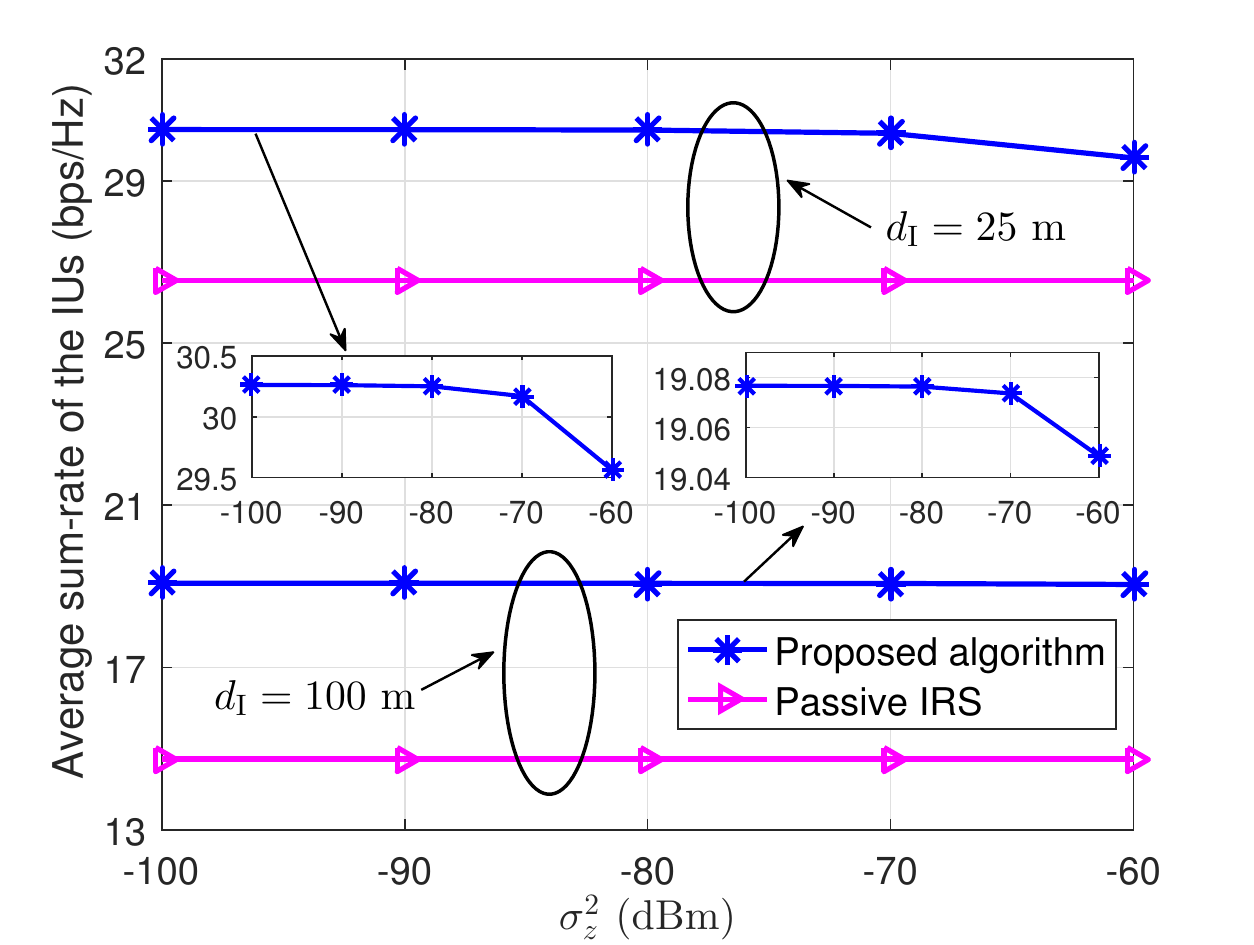}
		\caption{Average sum-rate of the IUs versus the noise varia-\\nce of the active IRS.}
		\label{fig:WSR_vs_noise}
	\end{minipage}%
	\begin{minipage}[t]{0.5\linewidth}
		\centering
		\includegraphics[scale=0.63]{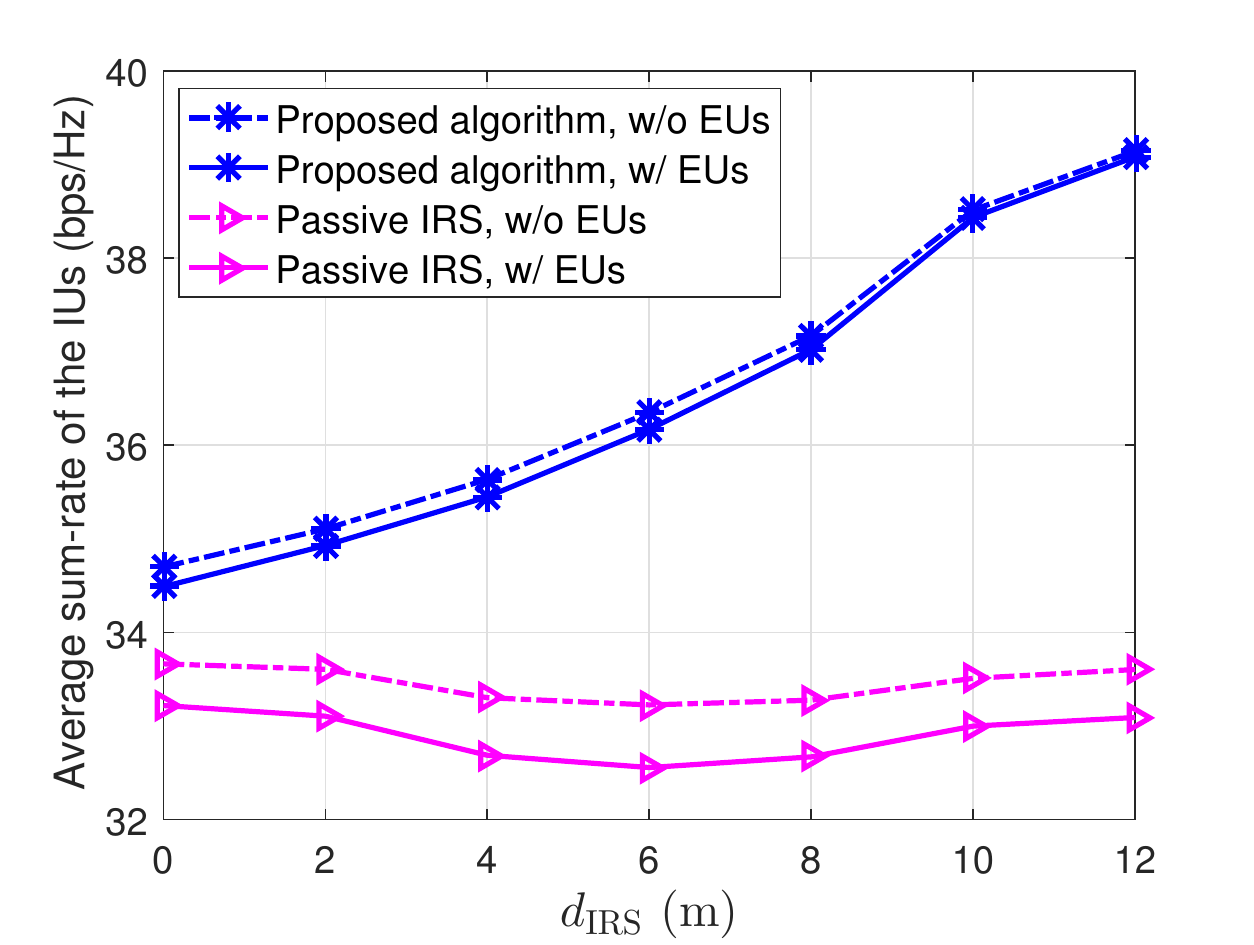}
		\caption{Average sum-rate of the IUs versus the y-axis coordinate value of the IRS.}
		\label{fig:WSR_vs_dIRS}
	\end{minipage}
	\vspace{-3mm}
\end{figure}

Finally, to gain more insights, we consider a setup where $d_{\rm I} = d_{\rm E} = 12$ m, i.e., the EUs and the IUs are randomly and uniformly located in the same disk centered at $(3, 12, 0)$ with radius $2$ m. In Fig. \ref{fig:WSR_vs_dIRS}, we compare the average sum-rate achieved by the proposed algorithm and the passive IRS scheme for both cases with and without the EUs versus the y-axis coordinate value of the IRS when $E = 1$ $\mu$W. The other simulation parameters are the same as those in Fig. \ref{fig:WSR_vs_E}. In addition to observations similar to those in Fig. \ref{fig:WPT_WSP_vs_dIRS}, we can see that the existence of the EUs brings a much smaller performance degradation to the proposed algorithm than to the scheme employing a passive IRS, especially when the active IRS is deployed in the proximity of the users. This result further shows the superiority of deploying an active IRS in enhancing the performance of SWIPT systems with both EUs and IUs coexisting. 

\vspace{-2mm}
\section{Conclusion}\label{Section_conclusion}
In this paper, we studied the weighted sum-power and sum-rate maximization problems in an active IRS-assisted SWIPT system. Specifically, the first problem aimed to maximize the weighted sum-power harvested by the EUs while meeting the specified SINR targets at the IUs, and the second problem was intended to maximize the weighted sum-rate of the IUs while satisfying the EH requirements at the EUs. In both problems, the transmit precoder at the AP and the reflection-coefficient matrix at the IRS were jointly optimized. Interestingly, it was rigorously proved that there is no loss of optimality in removing dedicated energy beams in the SDR reformulations of both optimization problems. Based on these results, efficient suboptimal algorithms were proposed for the resulting problems. Numerical results verified that, compared with the benchmark scheme using a passive IRS, the proposed designs with an active IRS are able to significantly enhance the performance of both the EUs and the IUs. Useful insights on the appropriate deployment of an active IRS were also identified, providing helpful guidance for the practical design and implementation.                       

\appendices
\section{Proof of Theorem \ref{theo1}} \label{Appen_A}
To prove Theorem \ref{theo1}, we aim to show that (P1-SDR1) shares the same optimal value with the following problem (P1-SDR2) and that an optimal solution satisfying ${\rm rank}\left(\bm W_i^*\right) = 1$, $\forall i\in\mathcal{K_I}$ exists for (P1-SDR2):
\begin{subequations}
	\begin{eqnarray}
	\hspace{-1cm}\text{(P1-SDR2)}: &\underset{\left\lbrace \boldsymbol W_i\in\mathbb H^M\right\rbrace,\mathbf \Theta}{\max}& \sum_{i\in\mathcal {K_I}} {\rm tr}\left(\bm S \bm W_i \right) + \sum_{j\in\mathcal {K_E}}\alpha_j\sigma_z^2\left\|\bm g_{r,j}^H\mathbf \Theta\right\|^2 \\
	&\text{s.t.}& \frac{{\rm tr}\left(\bm h_i\bm h_i^H\bm W_i \right) }{\gamma_i} - \sum_{k\in \mathcal {K_I}\backslash\{i\}}{\rm tr}\left(\bm h_i\bm h_i^H\bm W_k \right) - \bar \sigma_i^2 \geq 0,  \forall i\in\mathcal {K_I}, \label{SDR2_1}\\
	&& \sum_{i\in\mathcal {K_I}}{\rm tr}\left(\bm W_i \right) \leq P_{\text{A}}, \ \sum_{i\in\mathcal {K_I}}{\rm tr}\left(\bm C\bm W_i \right) \leq \bar P_{\text{I}}, \label{SDR2_3}\\
	&& \bm W_i \succeq \bm 0, \forall i\in\mathcal {K_I}. \label{SDR2_4}
	\end{eqnarray}
\end{subequations}
\normalsize Denote $\zeta_1^*$ and $\zeta_2^*$ as the optimal objective values of (P1-SDR1) and (P1-SDR2), respectively. Suppose that $\{ \{\hat {\bm W}_i\}, \hat {\bm W}_{\mathrm E}, \hat {\mathbf \Theta}\} $ is an arbitrary optimal solution to (P1-SDR1) corresponding to $\zeta_1^*$. Obviously, $\zeta_1^* \geq \zeta_2^*$ since (P1-SDR2) is a special case of (P1-SDR1) with $\bm W_{\mathrm E} = \mathbf 0$. Next, we prove that $\zeta_1^* \leq \zeta_2^*$ also holds, building on the insight that by adding $\hat {\bm W}_{\mathrm E}$ into any $\hat {\bm W}_i$, we can construct a feasible solution to (P1-SDR2) that achieves the same objective value as $\zeta_1^*$. Specifically, let $\tilde {\bm W}_m = \hat{\bm W}_m + \hat{\bm W}_{\mathrm E}$ for any $m \in\mathcal {K_I}$ and $\tilde {\bm W}_i = \hat{\bm W}_i$, $\forall i\in\mathcal {K_I}\backslash\{m\}$. It is easy to verify that the constraints in \eqref{SDR2_3} and \eqref{SDR2_4} hold for the new solution set $\tilde{\Gamma} \triangleq \{\{ \tilde {\bm W}_i\}, \hat {\mathbf \Theta}\} $ and the objective value of (P1-SDR2) achieved by $\tilde{\Gamma}$ equals $\zeta_1^*$. Then, we show that $\tilde{\Gamma}$ also fulfills the constraints in \eqref{SDR2_1}. To this end, the following two cases are considered: 
\subsubsection{For $m \in\mathcal {K_I}$} We have
\begin{align}
& \frac{{\rm tr}\left(\bm h_m\bm h_m^H\tilde{\bm W}_m \right) }{\gamma_m} - \sum_{k\in \mathcal {K_I}\backslash\{m\}}{\rm tr}\left(\bm h_m\bm h_m^H\tilde{\bm W}_k \right) - \bar \sigma_m^2 \nonumber\\
= & \ \frac{{\rm tr}\left(\bm h_m\bm h_m^H\hat{\bm W}_m \right) }{\gamma_m} + \frac{{\rm tr}\left(\bm h_m\bm h_m^H\hat{\bm W}_E \right) }{\gamma_m}- \sum_{k\in \mathcal {K_I}\backslash\{m\}}{\rm tr}\left(\bm h_m\bm h_m^H\hat{\bm W}_k \right) - \bar \sigma_m^2 \nonumber\\
\overset{(a)}{\geq} & \ \frac{{\rm tr}\left(\bm h_m\bm h_m^H\hat{\bm W}_m \right) }{\gamma_m}  - {\rm tr}\left(\bm h_m\bm h_m^H\hat{\bm W}_E \right)-\sum_{k\in \mathcal {K_I}\backslash\{m\}}{\rm tr}\left(\bm h_m\bm h_m^H\hat{\bm W}_k \right) -\bar \sigma_m^2 \overset{(b)}{\geq} \ 0, 
\end{align}
\normalsize where $(a)$ uses $\hat{\bm W}_E \succeq \mathbf 0$ and $(b)$ holds due to \eqref{SINR_orig}. 

\subsubsection{For any $i\in\mathcal {K_I}\backslash\{m\}$} It follows that
\begin{align}
& \frac{{\rm tr}\left(\bm h_i\bm h_i^H\tilde{\bm W}_i \right) }{\gamma_i} - \sum_{k\in \mathcal {K_I}\backslash\{i\}}{\rm tr}\left(\bm h_i\bm h_i^H\tilde{\bm W}_k \right) - \bar \sigma_i^2 \nonumber\\
& = \frac{{\rm tr}\left(\bm h_i\bm h_i^H\tilde{\bm W}_i \right) }{\gamma_i} - \sum_{ k\in \mathcal {K_I}\backslash\{i,m\}}{\rm tr}\left(\bm h_i\bm h_i^H\tilde{\bm W}_k \right) - {\rm tr}\left(\bm h_i\bm h_i^H\tilde{\bm W}_m \right) - \bar \sigma_i^2 \nonumber\\
& = \frac{{\rm tr}\left(\bm h_i\bm h_i^H\hat{\bm W}_i \right) }{\gamma_i} - \sum_{k\in \mathcal {K_I}\backslash\{i\}}{\rm tr}\left(\bm h_i\bm h_i^H\hat{\bm W}_k \right) - {\rm tr}\left(\bm h_i\bm h_i^H\hat{\bm W}_E\right)  - \bar \sigma_i^2 \geq 0,
\end{align}
\normalsize where the inequality follows from \eqref{SINR_orig}. 

Based on the above, we have
\begin{align}
\frac{{\rm tr}\left(\bm h_i\bm h_i^H\tilde{\bm W}_i \right) }{\gamma_i} - \sum_{k\in \mathcal {K_I}\backslash\{i\}}{\rm tr}\left(\bm h_i\bm h_i^H\tilde{\bm W}_k \right) - \bar \sigma_i^2 \geq 0,  \forall i\in\mathcal {K_I},
\end{align}
\normalsize which indicates that constraint \eqref{SDR2_1} also holds for $\tilde{\Gamma}$. As a result, $\tilde{\Gamma}$ is a feasible solution to (P1-SDR2). For (P1-SDR2), since its objective value achieved by the feasible solution $\tilde{\Gamma}$ is equal to $\zeta_1^*$ and must not be greater than its optimal objective value $\zeta_2^*$, we have $\zeta_1^* \leq \zeta_2^*$. Since $\zeta_1^* \geq \zeta_2^*$, we have $\zeta_1^* = \zeta_2^*$. 

Furthermore, according to \cite[Theorem 3.2]{2010_Yongwei_Rank}, there always exists an optimal solution to (P1-SDR2) satisfying $\sum_{ i\in \mathcal {K_I}}\left({\rm rank}\left(\bm W_i^*\right)\right)^2 \leq K_{\rm I} + 2$ under any $\mathbf \Theta$. Meanwhile, for $\gamma_i > 0$, $\forall i\in\mathcal{K_I}$, there must be $\bm W_i^* \neq \mathbf 0$ or equivalently ${\rm rank}\left(\bm W_i^*\right) \geq 1$. Then, it follows that ${\rm rank}\left(\bm W_i^*\right) = 1$, $\forall i\in\mathcal{K_I}$ should exist for (P1-SDR2). Combing the above results completes the proof.  

\section{Proof of Theorem \ref{theo2}}\label{Appen_B}
We prove Theorem \ref{theo2} by showing that (P2-Eqv-SDR1) shares the same optimal value with the following problem (P2-Eqv-SDR2) and that there always exists an optimal solution to (P2-Eqv-SDR2) satisfying ${\rm rank}\left( \bm W_i^*\right) \leq 1, \forall i\in\mathcal{K_I'}$, where $\mathcal{K_I'} \subseteq \mathcal{K_I}$ and $\left|\mathcal{K_I'}\right| \geq K_{\rm I} - 1$:
\begin{subequations}\label{P2_Eqv_SDR2}
	\begin{eqnarray}
    \hspace{-2.2cm}\text{(P2-Eqv-SDR2)}: 
	&\underset{\substack{\left\lbrace \boldsymbol W_i\in\mathbb H^M\right\rbrace, \\\mathbf \Theta, \left\lbrace \rho_i, \tau_i\right\rbrace}}{\max}& \sum_{i\in\mathcal {K_I}}\mu_i \log_2e^{\left(\rho_i - \tau_i\right)}  \label{P2_Eqv_SDR2_obj}\\	
	&\text{s.t.}& \sum_{k\in\mathcal {K_I}}{\rm tr}\left(\bm h_i\bm h_i^H\bm W_k \right) + \bar \sigma_i^2 \geq e^{\rho_i}, \forall i\in\mathcal{K_I}, \label{P2_Eqv_SDR2_cons:slack1}\\
	&& \sum_{k\in \mathcal {K_I}\backslash\{i\}}{\rm tr}\left(\bm h_i\bm h_i^H\bm W_k \right) + \bar \sigma_i^2 \leq e^{\tau_i}, \forall i\in\mathcal{K_I}, \label{P2_Eqv_SDR2_cons:slack2}\\
	&& \sum_{i\in\mathcal {K_I}}{\rm tr}\left(\bm g_j\bm g_j^H\bm W_i \right) \geq \bar E_j, \forall j\in\mathcal{K_E}, \label{P2_Eqv_SDR2_cons:E}\\
	&&  \sum_{i\in\mathcal {K_I}}{\rm tr}\left(\bm W_i \right) \leq P_{\text{A}}, \ \sum_{i\in\mathcal {K_I}}{\rm tr}\left(\bm C\bm W_i \right) \leq \bar P_{\text{I}},\label{P2_Eqv_SDR2_cons:amp}\\
	&&  \bm W_i \succeq \bm 0, \forall i\in\mathcal {K_I}. \label{P2_Eqv_SDR2_cons:posi}
	\end{eqnarray}
\end{subequations}
\normalsize Particularly, the proof of the equivalence between (P2-Eqv-SDR1) and (P2-Eqv-SDR2) is similar to that of the equivalence between (P1-SDR1) and (P1-SDR2) given in Appendix \ref{Appen_A}. Therefore, the details are omitted due to the space limitation.  

The remaining part is to prove that (P2-Eqv-SDR2) has an optimal solution where at least $K_{\rm I} - 1$ beamforming matrices $\{\bm W_i^*\}$ satisfy ${\rm rank}\left( \bm W_i^*\right) \leq 1$. Let $\check{\Xi} \triangleq \{\{\check{\bm W}_i\}, \check{\mathbf \Theta}, \{\check{\rho_i}, \check{\tau_i}\}\} $ be an arbitrary optimal solution to (P2-Eqv-SDR2), where ${\rm rank}\left( \check{\bm W}_i\right) > 1$ holds for more than one $i\in\mathcal{K_I}$. Then, we construct $\Xi^* \triangleq \{\{\bm W_i^*\}, \mathbf \Theta^*, \{\rho_i^*, \tau_i^*\}\} $ from $\check{\Xi}$ with\footnote{For any $\check{\bm W}_i = \mathbf 0$, $i\in \mathcal {K_I}\backslash\{m\}$, let $\bm W_i^* = \mathbf 0$.}
\begin{align}
& \mathbf \Theta^* = \check{\mathbf \Theta}, \rho_i^* = \check {\rho}_i, \tau_i^* = \check {\tau}_i, \forall i\in \mathcal {K_I},\label{constru_1}\\
& \bm w_i^* = \left(\bm h_i^H\check{\bm W}_i\bm h_i\right)^{-1/2}\check{\bm W}_i\bm h_i, \bm W_i^* = \bm w_i^*(\bm w_i^*)^H, \forall i\in \mathcal {K_I}\backslash\{m\}, \label{constru_2}\\
& \bm W_m^* = \sum_{i\in\mathcal {K_I}}\check{\bm W}_i - \sum_{i\in\mathcal {K_I}\backslash\{m\}}\bm W_i^*, \label{constru_3}
\end{align}
\normalsize where $m$ can be any element in the set $\mathcal {K_I}$. It is clear that ${\rm rank}\left( \bm W_i^* \right) \leq 1$ and $\bm W_i^* \succeq \bm 0$, $\forall i\in \mathcal {K_I}\backslash\{m\}$. In the following, we show that $\Xi^*$ is also an optimal solution to (P2-Eqv-SDR2). 

First, for any $\bm D \in \mathbb C^{M\times M}$, it holds that
\begin{align}\label{tr_equ}
\sum_{i\in\mathcal {K_I}}{\rm tr}\left(\bm D\bm W_i^*\right) = \sum_{i\in\mathcal {K_I}\backslash\{m\}}{\rm tr}\left(\bm D\bm W_i^*\right) + {\rm tr}\left(\bm D\bm W_m^*\right)
= \sum_{i\in\mathcal {K_I}}{\rm tr}\left(\bm D\check{\bm W}_i\right).
\end{align}
\normalsize Then, it is easy to see that constraints \eqref{P2_Eqv_SDR2_cons:slack1}, \eqref{P2_Eqv_SDR2_cons:E}, and \eqref{P2_Eqv_SDR2_cons:amp} hold for $\Xi^*$. 

Next, for any $\bm q\in \mathbb C^{M\times1}$, we have
\begin{align}\label{minus_nonnega}
\bm q^H\left(\check{\bm W}_i - \bm W_i^*\right)\bm q & = \bm q^H\check{\bm W}_i\bm q - \left(\bm h_i^H\check{\bm W}_i\bm h_i\right)^{-1}\left|\bm q^H\check{\bm W}_i\bm h_i\right|^2 \nonumber\\
\geq & \hspace{1mm} \bm q^H\check{\bm W}_i\bm q - \left(\bm h_i^H\check{\bm W}_i\bm h_i\right)^{-1}\left(\bm h_i^H\check{\bm W}_i\bm h_i\right)\left(\bm q^H\check{\bm W}_i\bm q\right) = 0, \forall i\in \mathcal {K_I}\backslash\{m\},
\end{align}
where the inequality holds due to the Cauchy-Schwarz inequality. It then follows that $\check{\bm W}_i - \bm W_i^* \succeq \bm 0$, $\forall i\in \mathcal {K_I}\backslash\{m\}$. Subsequently, we have
\begin{align}
\bm W_m^* = \sum_{i\in\mathcal {K_I}}\check{\bm W}_i - \sum_{i\in\mathcal {K_I}\backslash\{m\}}\bm W_i^* = \check{\bm W}_m + \sum_{i\in\mathcal {K_I}\backslash\{m\}}\left( \check{\bm W}_i-\bm W_i^*\right) \succeq \bm 0. 
\end{align}
Therefore, constraint \eqref{P2_Eqv_SDR2_cons:posi} holds for $\Xi^*$. 

It remains to prove that $\Xi^*$ ensures constraint \eqref{P2_Eqv_SDR2_cons:slack2}. To show this, we consider the following two cases: 
\subsubsection{For $m \in \mathcal{K_I}$} In this case, it follows that 
\begin{align}\label{case_1}
\sum_{k\in \mathcal {K_I}\backslash\{m\}}{\rm tr}\left(\bm h_m\bm h_m^H\check{\bm W}_k\right) \!- \! \sum_{k\in \mathcal {K_I}\backslash\{m\}}{\rm tr}\left(\bm h_m\bm h_m^H\bm W_k^*\right) 
= \sum_{k\in \mathcal {K_I}\backslash\{m\}}\bm h_m^H\left(\check{\bm W}_k - \bm W_k^*\right) \bm h_m \geq 0,
\end{align}
where the inequality follows from \eqref{minus_nonnega}. 

\subsubsection{For any $i\in\mathcal {K_I}\backslash\{m\}$}
First of all, it holds that $\bm h_i^H\bm W_i^*\bm h_i = \bm h_i^H\bm w_i^*(\bm w_i^*)^H\bm h_i = \bm h_i^H\check{\bm W}_k\bm h_i$, $\forall i\in\mathcal {K_I}\backslash\{m\}$. This, together with \eqref{constru_3}, yields
\begin{align}\label{case_2}
& \sum_{k\in \mathcal {K_I}\backslash\{i\}}{\rm tr}\left(\bm h_i\bm h_i^H\check{\bm W}_k\right) - \sum_{k\in \mathcal {K_I}\backslash\{i\}}{\rm tr}\left(\bm h_i\bm h_i^H\bm W_k^*\right) = \sum_{k\in \mathcal {K_I}\backslash\{i\}}{\rm tr}\left(\bm h_i\bm h_i^H\check{\bm W}_k\right) \nonumber\\
& \hspace{1.6cm} -\left( \sum_{k\in \mathcal {K_I}\backslash\{i,m\}}{\rm tr}\left(\bm h_i\bm h_i^H\bm W_k^*\right) + \sum_{k\in\mathcal {K_I}}{\rm tr}\left(\bm h_i\bm h_i^H\check{\bm W}_k\right)  - \sum_{k\in\mathcal {K_I}\backslash\{m\}}{\rm tr}\left(\bm h_i\bm h_i^H\bm W_k^*\right)  \right)  \nonumber \\
& \hspace{1.6cm} = \sum_{k\in \mathcal {K_I}\backslash\{i\}}{\rm tr}\left(\bm h_i\bm h_i^H\check{\bm W}_k\right) - \sum_{k\in\mathcal {K_I}}{\rm tr}\left(\bm h_i\bm h_i^H\check{\bm W}_k\right) + {\rm tr}\left(\bm h_i\bm h_i^H\bm W_i^*\right) \nonumber\\
& \hspace{1.6cm} = - {\rm tr}\left(\bm h_i\bm h_i^H\check{\bm W}_i\right) + {\rm tr}\left(\bm h_i\bm h_i^H\bm W_i^*\right) = 0.
\end{align}
With \eqref{case_1} and \eqref{case_2}, we have
\begin{align}
\sum_{k\in \mathcal {K_I}\backslash\{i\}}{\rm tr}\left(\bm h_i\bm h_i^H\bm W_k^*\right)  + \bar \sigma_i^2 \leq \sum_{k\in \mathcal {K_I}\backslash\{i\}}{\rm tr}\left(\bm h_i\bm h_i^H\check{\bm W}_k\right) + \bar \sigma_i^2 \leq e^{\check y_i} = e^{y_i^*}, \forall i\in\mathcal{K_I}, 
\end{align}
namely constraint \eqref{P2_Eqv_SDR2_cons:slack2} holds for $\Xi^*$. 

Finally, we note that the objective values in \eqref{P2_Eqv_SDR2_obj} achieved by $\Xi^*$ and $\check{\Xi}$ are identical. 

With the derivation above, it is verified that $\Xi^*$ is an optimal solution to (P2-Eqv-SDR2), where ${\rm rank}\left( \bm W_i^* \right) \leq 1$ holds for no less than $\left( K_{\rm I}-1\right) $ $i\in\mathcal{K_I}$. Theorem \ref{theo2} is thus proved.

\bibliographystyle{IEEEtran}
\bibliography{SWIPT}

\end{document}